\documentclass[leqno]{siamart1116}
\usepackage[T1]{fontenc}
\usepackage[utf8]{inputenc}

\numberwithin{theorem}{section}

\usepackage{amsmath,amsfonts,amssymb,dsfont,bbm,bbold,color,cite}
\usepackage{graphicx,mdwlist,subfigure,tikz}
\usepackage{algorithmic}
\usepackage[ruled,linesnumbered]{algorithm2e}
\usepackage{multirow,ctable} 
\usepackage[titletoc,toc,title]{appendix}

\def\uno{\mathbb 1}
\def\vol{\mathrm{vol}}
\def\cut{\mathrm{cut}}

\def\t{\text{\itshape{\textsf{T}}}}

\def\bar{\overline}
\def\sign{\mathrm{sign}}

\def\RR{\mathbbm R} 

\def\sp{\mathrm{span}}
\newcommand{\M}{\mathcal M}

\def\bar{\overline}
\def\tilde{\widetilde}
\def\emptyset{\varnothing}
\def\<{\left<} \def\>{\right>}
\def\fun{\tau}

\newtheorem{remark}[theorem]{Remark}

\title{Community detection in networks via  nonlinear modularity eigenvectors \thanks{\funding{The work of the authors has been supported by the ERC grant NOLEPRO. The work of F.T. has been partially supported by the Marie Curie Individual Fellowship MAGNET.}}}

\author{Francesco Tudisco\thanks{Department of Mathematics and Statistics, University of Strathclyde, G11XH Glasgow, UK (\email{f.tudisco@strath.ac.uk})} \and Pedro Mercado\thanks{Department of Mathematics and Computer Science, Saarland University, 66123 Saarbr\"ucken, Germany (\email{pedro@cs.uni-saarland.de}, \email{hein@cs.uni-saarland.de})} \and Matthias Hein\footnotemark[3]}

\begin{document} 
\maketitle

\begin{abstract}
Revealing a community structure in a network or dataset is a central problem arising in many scientific areas. The modularity function $Q$ is an established measure quantifying the quality of a community, being identified as a set of nodes having high modularity.  In our terminology, a set of nodes with positive modularity is called a \textit{module} and a set that  maximizes $Q$ is thus called \textit{leading module}. Finding a leading module in a network is an important task,  however the dimension of real-world problems makes the maximization of $Q$ unfeasible. This poses the need of approximation techniques which are typically based on a linear relaxation of $Q$, induced by the spectrum of the modularity matrix $M$. In this work we propose a nonlinear relaxation which is instead based on the spectrum of a nonlinear modularity operator $\M$. We show that extremal eigenvalues of $\M$ provide an exact relaxation of the modularity measure $Q$, in the sense that the maximum eigenvalue of $\M$ is equal to the maximum value of $Q$, however at the price of being more challenging to be computed than those of $M$. Thus we extend the work made on nonlinear Laplacians, by proposing a computational scheme, named \textit{generalized RatioDCA}, to address such extremal eigenvalues. We show monotonic ascent and convergence of the method. We finally apply  the new method to several synthetic and real-world data sets, showing both effectiveness of the model and performance of the method.
\end{abstract}



\pagestyle{myheadings}
\thispagestyle{plain}


\section{Introduction}
This paper is concerned with the problem of finding leading communities in a network. A community is roughly defined as a set of nodes being highly connected inside and poorly connected with the rest of the graph. Identifying important communities in a complex network is a highly relevant problem which has applications in many disciplines, such as computer science, physics, neuroscience, social science, biology, and many others, see e.g.\ \cite{newman2010networks, estrada2012structure, schaeffer2007graph, traud2011comparing}. 

In order to address this problem from the mathematical point of view one needs a quantitative definition of what a community is. To this end several merit functions have been introduced
in the recent literature \cite{fortunato2010community}. 
A very popular and successful idea is based on the concept of modularity introduced by Newman and Girvan in \cite{newman2004finding}. 

The modularity measure of a set of  nodes~$A$ in a graph $G=(V,E)$ quantifies the difference between the actual and expected weight of edges in $A$, if edges were placed at random according to a random \textit{null model}. A subgraph $G(A)$ is then identified as a community if the modularity measure of $A$ is ``large enough''. 

The modularity-based community detection problem thus boils down to a combinatorial optimization problem, that is reminiscent of another famous task known as graph partitioning. Graph partitioning  can be roughly described as 
the problem of finding a $k$-partition of the set of vertices of $G$, where $k$ is a \textit{given} number of disjoint sets to be identified.

Modularity-based community detection does not prescribe the number
of  subsets into which the network is divided, and it is generally assumed that the graph is intrinsically structured into groups that are 
delimited  to some extent. The main objective is to reveal the presence and the consistency of such groups. 

As modularity-based community detection is known to be NP-hard~\cite{brandes2008modularity}, different strategies have been proposed to compute an approximate solution. Linear relaxation approaches are based on the spectrum of specific matrices (as the modularity matrix or the Laplacian matrix) and have been been widely explored and applied to various research areas, see e.g. \cite{fasino2014algebraic,pmlr-v84-mercado18a,newman2006finding,shen2010spectral}. Computational heuristics have been developed for optimizing directly the discrete quality function (see e.g. \cite{porter2009communities,lancichinetti2009community}), including for example greedy algorithms \cite{clauset2004finding}, simulated annealing \cite{guimera2004modularity} and  extremal optimization \cite{duch2005community}. Among them, the locally greedy algorithm known as Louvain method \cite{blondel2008fast} is arguably the most popular one. 
In recent years, and mostly in the context of graph partitioning, nonlinear relaxation approaches have been proposed  (see for instance \cite{bresson2013multiclass,buhler2009spectral,hein2010inverse,tudisco2018core}). In the context of community detection, a nonlinear relaxation based  on the
Ginzburg-Landau functional  is considered for instance in \cite{hu2013method,boyd2017simplified}, where it is shown to be $\Gamma$-convergent
to the discrete modularity optimum.

In this paper, we propose two nonlinear relaxations of different modularity set functions, and prove them to be exact, in the sense that the maximum values of our proposed nonlinear relaxations are equal to the maximum of the corresponding modularity set functions.
More precisely, we introduce two nonlinear relaxations that are based on a nonlinear modularity operator $\M:\RR^n\to\RR^n$.
 We associate to $\M$ two different Rayleigh quotients, inducing two different notions of eigenvalues and eigenvectors of $\M$ and we prove two Cheeger-type results for $\M$ that  show that the maximal eigenvalues of $\M$ associated to such Rayleigh quotients coincide with the maxima of two different modularity measures of the graph.
Interestingly enough we observe that the modularity matrix completely overlooks the difference between these two modularity measures, which instead $\M$ allows to address individually. 

Although nonlinearity generally prevents us to compute the eigenvalues of $\M$, the optimization framework proposed in \cite{hein2011beyond} allows for an algorithm addressing the minimization of positive valued Rayleigh quotients. As the Rayleigh quotients we associate to $\M$ attain positive and negative values, here we extend that method to a wider class of ratios of functions, proving monotonic descent and convergence to a nonlinear eigenvector. 

The paper is organized as follows:  Section \ref{sec:modularity_measure} gives an overview of the concept of modularity measure, modularity matrix and the Newman's spectral method for community detection, as proposed in \cite{newman2004finding}. In Section \ref{sec:nonlinear_modularity} we define the nonlinear modularity operator $\M$ and the associated Rayleigh and dual Rayleigh quotients. We show that both ensure an exact relaxation of suitable modularity-based combinatorial optimization problems on the graph. In Section \ref{sec:nonlinear_spectral_method} we propose a nonlinear spectral method for community detection in networks through the eigenvectors of the nonlinear modularity and, finally, in Section \ref{sec:experiments} we show extensive results on synthetic and real-world networks highlighting the improvements that nonlinearity ensures over the standard linear relaxation approach. 

\subsection{Notation}
Throughout this paper we assume that an undirected graph $G =(V,E)$ is given, with the following properties: $V$ is the vertex set equipped with the positive measure $\mu:V\to \RR_+$; $E$ is the edge set equipped with positive weight function $w:E\to \RR_+$. The symbol $\RR_+$ denotes the set of positive numbers. The vertex set $V$ is everywhere identified with $\{1, \dots, n\}$. We denote by $\<\cdot,\cdot\>_\mu$ the weighted scalar product $\<x,y\>_\mu=\sum_i \mu_i x_iy_i$. Similarly, for $p\geq 1$ we let $\|x\|^p_{p,\mu}=\sum_{i} \mu_i |x_i|^p$ be the weighted $\ell^p$ norm on $V$.

Given two subsets $A, B \subseteq V$, the set of edges between nodes in $A$ and $B$ is denoted by $E(A, B)$. When $A$ and $B$ coincide we use the short notation $E(A)$. The overall weight of a set is the sum of the weights in the set, thus for $A, B\subseteq V$, we write
$$\mu(A)=\sum_{i\in A}\mu_i, \qquad w(E(A,B)) = \sum_{ij \in E(A,B)}w(ij)\, .$$
Special notations are  reserved to the case where $B$ is the whole vertex set. Precisely, $w(E(\{i\}, V))=d_i$ is the degree of the node $i$, and $w(E(A,V))=\vol(A)=\sum_{i \in A}d_i$ the volume of the set $A$.

For a subset $A\subseteq V$ we write $\bar A$ to denote the complement $V\setminus A$ and we let $\uno_A\in\RR^n$ be the characteristic vector $(\uno_A)_i=1$ if $i \in A$ and $(\uno_A)_i=0$ otherwise.

\section{Modularity measure}\label{sec:modularity_measure}

A central problem in graph mining is to look for quantitative definitions of community. 
Although there is no universally accepted definition and a variety of merit functions have been proposed in recent literature, the global definition based on the modularity quality function proposed by Newman and Girvan \cite{newman2004finding} is an effective and very  popular one  \cite{fortunato2010community}. Such measure is based on the assumption that $A\subseteq V$ is a community of nodes if the induced subgraph $G(A)=(A,E(A))$ contains more edges than expected, if edges were placed at random according to a random graph model $\mathcal G_0$ (also called null-model).

Let $G_0=(V_0,E_0)$  be the expected graph of the random ensemble $\mathcal G_0$, with weight measure $w_0:E_0\to\RR_+$. The definition of modularity $Q(A)$ of $A\subseteq V$,  is as follows 
\begin{equation}\label{eq:Q(A)}
 Q(A) = w(E(A))- w_0(E_0(A))\, , 
\end{equation}
so that $Q(A)>0$ if the actual weight of edges in $G(A)$ exceeds the expected one in $G_0(A)$. A set of nodes $A$ is a community if it has positive modularity, and the associated subgraph $G(A)$ is called a module.  
A number of different null-models and variants of the modularity measure have been considered in recent literature, see e.g. \cite{fasino2016generalized, reichardt2006statistical, multires2, multires3}. 

An alternative formulation relates with a normalized version of the modularity, where the measure $\mu(A)$ of the set $A$ is used as a balancing function, for different choices of the measure $\mu$. 
We define the normalized modularity $Q_\mu(A)$ of $A\subseteq V$ as follows 
\begin{equation}
 Q_\mu(A) = Q(A)/\mu(A)\, .
\end{equation}
As we discuss in Section \ref{sec:experiments}, the use of such normalized version can help to identify small group of nodes as important communities in the graph, whereas it is known that the standard (unnormalized) measure tends to overlook small groups \cite{fortunato2007resolution}.

The definition of modularity of a subset is naturally extended to the measure of the modularity of a partition of $G$, by simply looking at the sum of the modularities: given a partition $\{A_1, \dots, A_k\}$ of $V$, its modularity and normalized modularity are defined respectively by 
$$q(A_1, \dots, A_k)=\frac 1 {\mu(V)}\sum_{i=1}^k Q(A_i), \qquad \text{and} \qquad q_\mu(A_1, \dots, A_k)=\sum_{i=1}^k Q_\mu(A_i)\, .$$
Clearly the normalization factor $1/\mu(V)$ does not affect the community structure and is considered here for compatibility with previous works. When the partition consists of only two sets $\{A,\bar A\}$ we use the shorter notation $q(A)$ and $q_\mu(A)$ for $q(A,\bar A)$ and $q_\mu(A,\bar A)$, respectively. 

The definition and effectiveness of the modularity measure \eqref{eq:Q(A)} highly depends on the chosen random model $\mathcal G_0$. A very popular and successful one, considered originally by Newman and Girvan in \cite{newman2004finding}, is based on the Chung-Lu random graph (see f.i. \cite{chung2006complex,arcolano2012moments,prvzulj2006modelling}) and its weighted variant \cite{fasino2014algebraic}.
For the sake of completeness, we recall hereafter the definition of  weighted Chung-Lu model.
\begin{definition}   \label{def:wCL}
Let $\delta = (\delta_1,\ldots,\delta_n)^\t > 0$, and
let $X(p)$ be a nonnegative random variable parametrized
by the scalar parameter $p\in[0,1]$, whose expectation is 
$\mathbbm{E}(X(p)) = p$. 
We say that a graph $G = (V,E)$ with weight function $w$ follows the 
{\em $X$-weighted Chung-Lu random graph model} $\mathcal{G}(\delta,X)$ 
if, for all $i,j\in V$, 
$w(ij)$ are independent random variables
distributed as $X(p_{ij})$ where
$p_{ij} = \delta_i\delta_j/\sum_{i=1}^n\delta_i$.
\end{definition}

The unweighted model coincides with the special case
of $\mathcal{G}(\delta,X)$ where $X(p)$ is the Bernoulli trial 
with success probability $p$. On the other hand, if $X(p)$
has a continuous part, then $\mathcal{G}(\delta,X)$
may contain graphs with generic weighted edges.
In any case, as in the original Chung-Lu model, 
if $G$ is a random graph drawn from $\mathcal{G}(\delta,X)$ then 
the expected degree of node $i$ is
$\mathbbm{E}(d_i) = \delta_i$. 

Given the degree sequence $d=(d_1,\dots, d_n)$ of the actual network $G=(V,E)$,  we assume from now on that the null-model $\mathcal G_0$ follows the weighted Chung-Lu random graph model $\mathcal G(\delta,X)$ above, with $\delta = d$. Note that, under this assumption, the modularity measure \eqref{eq:Q(A)} becomes $Q(A) = w(E(A))- {\vol(A)^2}/{\vol(V)}$ 
and we have, in particular, $Q(A)=Q(\bar A)$, for any $A\subseteq V$. 

The main contributions we propose in this work deal with the leading module problem, that is the problem of finding a subset $A\subseteq V$ having maximal modularity. Due to the identity $Q(A)=Q(\bar A)$, such problem coincides with finding the bi-partition $\{A, \bar A\}$ of the vertex set, having maximal modularity. Note that, for the special case of partitions consisting of two sets, the corresponding modularity and normalized modularity set functions are
\begin{equation}\label{eq:q(A)}
q(A)=\frac 2 {\mu(V)} Q(A), \qquad \text{and} \qquad q_\mu(A)= \mu(V) \frac{Q(A)}{\mu(A)\mu(\bar A)}\, . 
\end{equation}

\subsection{The modularity matrix and the spectral method}
Looking for a leading module is a major task in community detection which coincides with the discovery of an optimal bi-partition of $G$ into communities, in terms of modularity. 
This problem is equivalent to maximizing the modularity and normalized modularity through
the set functions $q$ and $q_\mu$, respectively,
over the possible subsets of $V$, 
namely computing the quantities
\begin{equation}\label{eq:cut_modularities}
 q(G) = \max_{A \subseteq V}q(A), \qquad q_\mu(G)=\max_{A \subseteq V}q_\mu(A)\, .
\end{equation}

As both $q(G)$ and $q_\mu(G)$ are NP-hard optimization problems \cite{brandes2008modularity}, a globally optimal solution for large graphs is out of reach.  One of the best known techniques for an approximate solution to these problems -- typically referred to as ``spectral method'' -- relates with the modularity matrix, and its leading eigenpair. Let $d$ be the vector of the degrees of the graph, the normalized modularity matrix of $G$, with vertex measure $\mu$, is defined as  follows
\begin{equation*}\label{eq:linearM}
(M)_{ij} = \frac 1 {\mu_i}\left( w(ij) - \frac{d_i d_j}{\vol(V)}\right), \qquad \text{for }i,j =1,\dots, n\, .
\end{equation*}

Note that the term $d_i d_j/\vol(V)$ is the $(i,j)$ entry of the one-rank  adjacency matrix of the expected graph of a random ensemble following the weighted Chung-Lu random model. 

The spectral method roughly selects a bi-partition of the vertex set $V$ accordingly with the sign of the elements in an eigenvector $x$ of $M$, associated with its largest eigenvalue $\lambda_1(M)$. It is proved in \cite{fasino2016generalized} that if $\tilde d =(d_1/\sqrt{\mu_1}, \dots, d_n/\sqrt{\mu_n})$ is not an eigenvector of $M$, then $\lambda_1(M)$ is a simple eigenvalue and thus $x$ is uniquely defined.
If $\lambda_1(M)>0$, one computes $x$ such that $Mx=\lambda_1(M)x$, then the vertex set $V$ is partitioned into $A_+ = \{i \in V : x_i \geq t_*\}$  and $\bar {A_+}$, being $t_* =\arg\max_t q(\{i\in V : u_i\geq t\})$.  If $\lambda_1(M) = 0$, the graph is said algebraically indivisible, i.e.\ it resembles no community structure (see e.g.\ \cite{newman2006finding,fasino2014algebraic}).  The motivations behind this technique are based on a relaxation argument, that we discuss in what follows.

The Rayleigh quotient of $M$ is the real valued function 
$$r_M(x) = \frac{\<x, Mx\>_\mu}{\, \,\,  \|x\|_{2,\mu}^2}\, .$$
As the matrix $M$ is symmetric with respect the weighted scalar product $\<\cdot, \cdot\>_\mu$, its eigenvalues can be characterized as variational values of $r_M$. In particular, if  the eigenvalues of $M$ are enumerated in descending order, then $\lambda_1(M)$ is the global maximum of $r_M$, 
\begin{equation}\label{eq:lambda1(M)}
\lambda_1(M) = \max_{x \in \RR^n} \, r_M(x)\, .
\end{equation}
The quantity $q(G)$ can be rewritten in terms of $r_M$, thus in terms of $M$. Consider the binary vector $v_A = \uno_A - \uno_{\bar A}$. Using the identities $\uno_{\bar A} = \uno-\uno_A$,  $M\uno=0$ and $\|v_A\|_{2,\mu}^2 = \mu(V)$, we get $\<v_A, Mv_A\>_\mu=4Q(A)$, thus
\begin{equation}\label{eq:linear_relaxation_Q}
 q(G)=\max_{A \subseteq V}\frac{2\, Q(A)}{\mu(V)} = \frac{1}{2}\, \max_{A\subseteq V}\, r_M(v_A) = \frac{1}{2}\, \max_{ x \in \{-1, 1\}^n} r_M(x)\, .
\end{equation}
Computing the global optimum of $r_M$ over $\{-1, 1\}^n$ is NP-hard. However, this maximum can be approximated by dropping the binary constraint  on $x$ and, thus, transforming 
the problem into the eigenvalue problem \eqref{eq:lambda1(M)}, which can be easily solved. This observation is  one of the main motivations of the spectral method based on the modularity matrix $M$ and its largest eigenvalue $\lambda_1(M)$, whereas the main drawback of this approach is that, in general,   the eigenvalue $\lambda_1(M)$ can arbitrary differ from the actual modularity $q(G)$. 

From Equation~\eqref{eq:linear_relaxation_Q} we can see   that $r_M$ coincides with $q$ when evaluated on binary vectors $x\in \{-1,1\}^n$. For this reason and the fact that $\max_{x\in \RR^n}r_M(x)$ coincides with an eigenvalue of the linear operator $M$ we say that $r_M$ is a \textit{linear relaxation} of $q$.

Before concluding this section we would like to point out another drawback of the linear relaxation approach which, to our opinion, is often overlooked: as we will show in Section \ref{sec:experiments}, the solutions of $q(G)$ and $q_\mu(G)$ are in general far from being the same, however the linear relaxation approach in principle ignores such difference. In fact, the linear relaxation $r_M$ of the modularity set function $q$ is also a linear relaxation of the normalized modularity set function $q_\mu$. We show such observation via the following
\begin{proposition}\label{pro:linear_relaxation_Qmu}
 If the largest eigenvalue $\lambda_1(M)$ of $M$ is positive, then $r_M$ is a linear relaxation of both $q$ and $q_\mu$.
\end{proposition}
\begin{proof}
 We already observed that $r_M$ coincides with $q$ on the set of binary vectors $x\in \{-1,1\}^n$. 
 A similar simple argument is used for $q_\mu$. Consider the vector $w_A = \uno_A - \frac{\mu(A)}{\mu(V)}\uno$. Since $\mu(\bar A )= \mu(V)-\mu(A)$ we get $\|w_A\|_{2,\mu}^2=\frac{\mu(A)\mu(\bar A)}{\mu (V)}$ and $r_M(w_A)=q_\mu(A,\bar A)$. Note that $\<w_A, \uno\>_\mu =0$, thus 
\begin{equation}\label{eq:linear_relaxation_Qmu}
 q_\mu(G)=\max_{A \subseteq V}\frac{\mu(V)\, Q(A)}{\mu(A)\, \mu(\bar A)} = \, \max_{A\subseteq V}\, r_M(w_A) =  \max_{ x \in \{-a, b\}^n, \, \<x,\uno\>_\mu=0} r_M(x)\, .
\end{equation}
 Therefore, $r_M$ and $q_\mu$ coincide on the set of binary vectors $x\in \{-a,b\}^n$ such that  $\<x,\uno\>_\mu=0$ (for suitable $a,b>0$).
 As $M$ has a positive eigenvalue by assumption, dropping the binary constraint $x \in \{-a,b\}^n$ and recalling that $\uno \in \ker(M)$, we get $ \max_{ x \in \RR^n, \, \<x,\uno\>_\mu=0}r_M(x)=\max_{x\in\RR^n}r_M(x)=\lambda_1(M) $. 
\end{proof}

\section{Tight nonlinear modularity relaxation}\label{sec:nonlinear_modularity}
In this section we introduce a nonlinear modularity operator $\M$, through a natural generalization of the modularity matrix $M$. To this operator we associate a Rayleigh quotient and a dual Rayleigh quotient to which naturally correspond a notion of nonlinear eigenvalues and eigenvectors. 
We use the new Rayleigh quotients to derive nonlinear relaxations of the modularity $q$ and the normalized modularity $q_\mu$ set functions, respectively. Moreover, unlike the standard linear relaxation,  we show that such relaxations are tight, that is we prove a Cheeger-type result showing that certain eigenvalues of $\M$ coincide with the graph modularities \eqref{eq:cut_modularities}. 

\subsection{Nonlinear modularity operator}
The nonlinear modularity operator we are going to define is related with the Clarke's subdifferential $\partial$  (see \cite{clarke1990optimization} e.g.). We recall that, for $f:\RR^n\to\RR$ Lipschitz around $x\in \RR^n$, the subdifferential of $f$ at $x$ is defined as the following subset of $\RR^n$
$$
\partial f(x) = \left\{y\in\RR^n : \<y,v\>\leq \limsup_{z\to x,t\to 0}\frac{f(z+tv)-f(z)}{t}, \quad \text{for all }v\in\RR^n \right\}\, .
$$

The subdifferential of the one norm  $f(x) = \|x\|_1$ and the infinity norm $f(x) = \|x\|_\infty$ are of particular importance of us. For these particular functions explicit expressions for $\partial f(x)$ are available. We recall them below in \eqref{eq:subdiff_norm1} and \eqref{eq:subdiff_norminfty}, respectively. 


As the absolute value is not differentiable at zero, the subdifferential of the $1$-norm is the set valued map $\Phi$ defined by 
\begin{equation}\label{eq:subdiff_norm1}
x\mapsto \Phi(x) =\left\{y\in \RR^n : 
\begin{array}{l}
y_i = \sign(x_i) \text{ if } x_i\neq 0,\\ 
y_i\in [-1,1] \text{ if } x_i=0 
\end{array} \right\} \, ,
\end{equation}
%
where $\sign(x_i)=1$ if $x_i>0$ and $\sign(x_i)=-1$ if $x_i<0$.
Note that if $y \in \Phi(x)$ then any component of $y$ belongs to the image of the corresponding component of $x$. Precisely, $y\in \Phi(x)$ if and  only if $y_i\in \Phi(x_i)$ for all $i=1,\dots,n$.

In order to define the nonlinear modularity operator, let us first observe that, due to the identity  $\sum_{j=1}^n M_{ij}=0$, for $i=1, \dots, n$, the following formula holds for the modularity matrix $M$:
\begin{align*}
	(Mx)_i = \sum_{j=1}^n M_{ij}x_j - x_i \sum_{j=1}^n M_{ij}=\sum_{j=1}^n (-M)_{ij}(x_i-x_j)\, .
\end{align*}
This implies the following identity
\begin{align}\label{eq:quadratic_form_M}
	\<x, M x\>_\mu &= \sum_{i,j=1}^n \mu_i (-M)_{ij}x_i(x_i-x_j) 
	= \frac 12 \sum_{i,j=1}^n \mu_i(-M)_{ij}|x_i-x_j|^2\, ,
\end{align}
for any $x\in \RR^n$. Thus we define the nonlinear modularity operator as follows:
\begin{equation}\label{eq:nonlinearM}
 \M(x)_i =  \sum_{j=1}^n (-M)_{ij}\Phi(x_i-x_j),\quad i = 1,\dots, n\, .
\end{equation}
Note that, by definition, for any $y\in \M(x)$ we have
\begin{equation}\label{eq:TV}
\<x, y\>_\mu = \frac 12 \sum_{i,j=1}^n \mu_i(-M)_{ij}|x_i-x_j|\, .
\end{equation}

Since the right-hand side of \eqref{eq:TV} does not depend on the choice of the vector $y\in \M(x)$, we write $\<x,\M(x)\>_\mu$ to denote the quantity in \eqref{eq:TV}, in analogy with \eqref{eq:quadratic_form_M}. 
Note that $\<x,\M(x)\>_\mu$ and $\<x,M(x)\>_\mu$ coincide on binary vectors, for instance when $x\in \{-1,1\}^n$. 
Also note that
 $\<x,\M(x)\>_\mu$ is strictly related with the total variation of the vector $x$. More precisely, $\<x,\M(x)\>_\mu$ is the difference of two weighted total variations of $x$, as we will discuss with more detail in Section \ref{sec:exact_relaxation}. For completeness, we recall that the weighted total variation of $x\in \RR^n$ is the scalar function 
$$
|x|^\rho_{TV} = \sum_{i,j=1}^n\rho(i,j)|x_i-x_j|\, ,
$$
where $\rho(i,j)\geq 0$ are the nonnegative weights. 

We now consider two Rayleigh quotients associated with $\M(x)$, defined as follows
\begin{equation}\label{eq:Rayleigh_quotients}
 r_\M(x) = \frac{\<x,\M(x)\>_\mu}{\, \, \, \|x\|_{1, \mu}}, \qquad r_\M^*(x)=\frac{\<x,\M(x)\>_\mu}{\, \,\,  \|x\|_{\infty}}\, ,
\end{equation}
where $\|x\|_{1,\mu}=\sum_i\mu_i|x_i|$ and $\|x\|_\infty = \max_i |x_i|$. The functions $r_\M$ and $r_\M^*$ generalize the Rayleigh quotient $r_M$ of the linear modularity, and we will show in the next section that the global maxima of $r_\M^*$ and $r_\M$ provide an exact nonlinear relaxation of the modularity $q$ and normalized modularity $q_\mu$ set functions, defined in~\eqref{eq:q(A)}, respectively.
Here we show that the optimality conditions for $r_\M$ and $r_\M^*$ are related to a notion of eigenvalues and eigenvectors for the nonlinear modularity operator $\M$. We also briefly discuss the underlying mathematical reason why $r_M$ naturally generalizes into $r_\M$ and $r_\M^*$.

\subsection{Nonlinear modularity eigenvectors}
As for the 1-norm, we consider the subdifferential $\Psi$ of the infinity norm  $x\mapsto \|x\|_{\infty}$. For a vector $x\in \RR^n$, let $m_1, \dots, m_k$ be the indices such that $|x_{m_i}|=\|x\|_\infty$, then 
the subdifferential of the infinity norm is the set valued map $\Psi$ defined by
\begin{equation}\label{eq:subdiff_norminfty}
x\mapsto \Psi(x)= 
\mathrm{Conv}\{\sigma_1\uno_{m_1}, \dots, \sigma_k\uno_{m_k}\}\, ,
\end{equation}
where, for $i=1,\dots,k$, 
$\sigma_i=\sign(x_{m_i})$
and $\mathrm{Conv}$ denotes the convex hull. 

To the subdifferentials $\Phi$ and $\Psi$ correspond a notion of eigenvalue and eigenvector of $\M$
\begin{definition} \label{def:eig}
We say that $\lambda$ is a nonlinear eigenvalue of $\M$ with eigenvector $x$ if either $0\in \M(x)-\lambda \Phi(x)$ or $0\in \M(x)-\lambda\Psi(x)$. 
\end{definition}

We have
\begin{proposition}
 Let $x$ be a critical point of $r_\M$, then $x$ is a nonlinear eigenvector of $\M$ such that $0 \in \M(x)-\lambda \Phi(x)$ with $\lambda = r_\M(x)$. Similarly, if $x$ is a critical point of $r_\M^*$, then $x$ is a nonlinear eigenvector of $\M$ such that  $0 \in \M(x)-\lambda \Psi(x)$ with $\lambda = r_\M^*(x)$.
\end{proposition}
\begin{proof}
 Let $\partial$ denote the subdifferential. A direct inspection reveals that  $\partial \|x\|_{1,\mu} = D_\mu \Phi(x)$, where $D_\mu$ is the diagonal matrix $(D_\mu)_i = \mu_i$ and $\Phi(x)$ is the vector with components $\Phi(x)_i = \Phi(x_i)$. 
 Using the chain rule for $\partial$ (see e.g. \cite{clarke1990optimization}) we get
 \begin{align*}
  \partial\,  r_\M(x) &\subseteq \frac{1}{\|x\|_{1,\mu}^2}\big\{ \|x\|_{1,\mu} \, \, \partial \<x, \M(x)\>_\mu - \<x, \M(x)\>_\mu \partial \|x\|_{1,\mu} \big\}\\
  & = \frac{1}{\|x\|_{1,\mu}} \big\{D_\mu \M(x) - r_\M(x)\, D_\mu \Phi(x)\big\}
 \end{align*}
Therefore $0 \in \partial \, r_\M(x)$ implies $0 \in \M(x)-r_\M(x)\,  \Phi(x)$. As $\partial \|x\|_\infty = \Psi(x)$, a similar computation shows the proof for $r_\M^*$. 
\end{proof}

Thus critical points and critical values of $r_\M$ and $r_\M^*$ satisfy generalized eigenvalue equations for the nonlinear modularity operator $\M$. 
Despite the linear case, where the eigenvalues of the modularity matrix $M$ coincide with the variational values of $r_M$, the number of eigenvalues of $\M$  defined by means of the Rayleigh quotients in \eqref{eq:Rayleigh_quotients} is much larger than just the set of variational ones. 
However in many situations the variational spectrum plays a central role, as for instance in the case of the nonlinear Laplacian \cite{tudisco2016nodal, chang2016spectrum, drabek2002generalization}. This work provides a further example: in what follows we consider  the dominant eigenvalues of $\M$, coinciding with suitable variational values of $r_\M$ and $r_\M^*$, we prove two optimality Cheeger-type results and we discuss how to use these eigenvalues to locate a leading module in the network  by means of a nonlinear spectral method. The task of multiple community detection can also be addressed by successive bi-partitions, as we discuss in Section \ref{sec:real-networks}. Advantages of the nonlinear spectral method over the linear one are highlighted  Section \ref{sec:experiments} where extensive numerical results are shown.

\subsection{On the relation between $r_\M^*$, $r_\M$ and $r_M$} We briefly discuss the mathematical reason why $r_M$ generalizes into $r_\M$ and $r_\M^*$. This gives further reasoning to the definition in \eqref{eq:Rayleigh_quotients}. To this end we suppose for simplicity that $\mu_i = 1$. Therefore $(-M)_{ij}=d_id_j/\vol(V)-w(ij)$ for all $i, j =1, \dots,n$. Given the graph $G=(V,E)$, consider the linear difference operator  $B:\RR^n \to \RR^{|E|}$ entrywise defined by $(Bx)_{ij}=x_i-x_j$, $ij\in E$, and let $w_M:E\to \RR$ be the real valued function $w_M(ij)=(-M)_{ij}/2$.  Then we can write 
$$\<x,Mx\>_\mu =  \<Bx,Bx\>_{w_M} =  \|Bx\|_{2,w_M}^2 = \sum_{ij\in E} w_M(ij) (Bx)_{ij}^2\, ,$$ 
where we use the compact notation $\|\cdot\|_{2,w_M}$, even though that quantity is not a norm on $\RR^{|E|}$, as $w_M$ attains positive and negative values. We have, as a consequence, $r_M(x) = (\|Bx\|_{2,w_M}/\|x\|_2)^2$. 
A natural generalization of such quantity is therefore given by 
$$
r_p(x)=\left(\frac{\|Bx\|_{p,w_M}}{\|x\|_p} \right)^p\, ,
$$
where, for $p\geq 1$ and $z \in \RR^{|E|}$, we are using the notation $\|z \|_{p,w_M}^p=\sum_{ij} w_M(ij)|z_{ij}|^p$. Clearly $r_M$ is retrieved from $r_p$ for $p=2$. Now, let $p^*$ be the H\"older conjugate of $p$, that is  the solution of the equation $1/p + 1/p^* = 1$. As $2^* = 2$, the quantity $r_M(x)$ is in fact a special case of $r_p^*(x) = \left(\|Bx\|_{p,w_M} / \|x\|_{p^*} \right)^p$ as well. The Rayleigh quotients in \eqref{eq:Rayleigh_quotients} are obtained by plugging $p=1$ into $r_p$ and $r_p^*$, respectively. Even though in this work we shall focus only on the case $p=1$, we believe that further investigations on $r_p$ and $r_p^*$ for different values of $p$ would be of significant interest.
Figure \ref{fig:relations} outlines this observation and the relation between the set valued functions $q$ and $q_\mu$ and the Rayleigh quotients $r_p$ and $r_p^*$, for the specific values  $p=1,2$. The next section gives further detail in this sense.

\subsection{Exact relaxation via nonlinear Rayleigh quotients}\label{sec:exact_relaxation}

From \eqref{eq:linear_relaxation_Q} and Proposition \ref{pro:linear_relaxation_Qmu} we deduce that the leading eigenvalue $\lambda_1(M)$ of the modularity matrix $M$ is an upper bound for both the quantities $q(G)$ and $q_\mu(G)$. This intuitively motivates the use of such eigenvalue and the corresponding eigenvectors to approximate the modularity of the graph. However $\lambda_1(M)$  is an approximation that can be arbitrarily far from the true value of the modularity. In particular, when $\mu=d$ is the degree vector, a Cheeger-type inequality showing a lower bound for $q_\mu(G)$ in terms of $\lambda_1(M)$ has been shown in \cite{fasino2016modularity}, whereas a lower bound for $q(G)$ is known only for regular graphs \cite{fasino2014algebraic}, the general case being still an open problem. 

In what follows we show that moving from the linear to the nonlinear modularity operator, allows to shrink the distance between the combinatorial quantities $q(G)$ and $q_\mu(G)$ defined in \eqref{eq:cut_modularities} and the spectrum of $\M$. More precisely, we show that the new Rayleigh quotients $r_\M^*$ and $r_\M$, as for $r_M$, coincide with the modularity and normalized modularity functions $q$ and $q_\mu$, respectively, on suitable set of binary vectors. However, unlike the  linear case,  we prove that the quantities 
\begin{equation}\label{eq:lambdas}
 \lambda_1(r_\M^*)=\max_{x\in \RR^n}r_\M^*(x), \qquad \lambda_1^\bot(r_\M)=\max_{x\in \RR^n, \, \<x,\uno\>_\mu=0}r_\M(x)
\end{equation}
coincide exactly with the modularity $q(G)$ and normalized modularity $q_\mu(G)$, respectively.
For these reasons we say that the functions $r_\M^*$ and $r_\M$ are \textit{exact nonlinear relaxations} of the modularity $q$ and normalized modularity $q_\mu$ set functions, respectively. The diagram in Figure \ref{fig:relations} summarizes these relaxation relations.

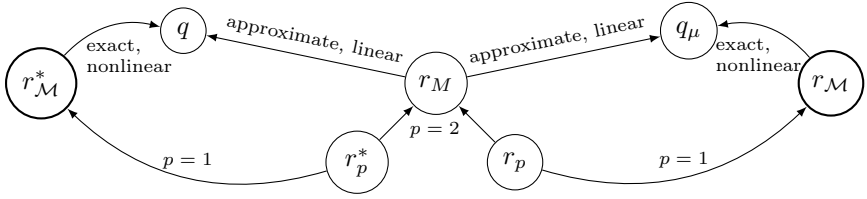
\begin{figure}
\centering
\begin{tikzpicture}[scale=.7]
\node[circle,draw=black,thick] (a) at (-.5,0)  {$r_\M^*$};
\node[circle,draw=black] (b) at (2.2,1)  {$q$};
\node[circle,draw=black] (c) at (7,0)  {$r_M$};
\node[circle,draw=black] (d) at (11.8,1)  {$q_\mu$};
\node[circle,draw=black,thick] (e) at (14.5,0) {$r_\M$};
\node[circle,draw=black] (f) at (5.5,-1.5) {$r_p^*$};
\node[circle,draw=black] (g) at (8.5,-1.5) {$r_p$};

\path[-latex] %
(a) edge[bend left] node[below]{\scriptsize$\qquad\begin{array}{l}\text{exact},\\ \text{nonlinear}\end{array}$} (b) %
(c) edge node[right, pos=.7,sloped,anchor=south,auto=false]{$\,$ \hspace{3em}\scriptsize approximate, linear} (b) %
(c) edge node[right, pos=.25,sloped,anchor=south,auto=false]{$\,$ \hspace{3em}\scriptsize approximate, linear} (d) %
(e) edge[bend right] node[below, pos=.6]{\scriptsize$\begin{array}{l}\text{exact},\\ \text{nonlinear}\end{array}$} (d) %
(f) edge node[right, pos=.3]{\scriptsize $\,\,\, p=2$} (c) %
(g) edge (c) %
(f) edge[bend left] node[above]{\scriptsize $p=1$} (a) %
(g) edge[bend right] node[above]{\scriptsize $p=1$} (e);
\end{tikzpicture}
\caption{This diagram summarizes how the relaxation functions $r_p$ and $r_p^*$ (for $p=1,2$) are related with $q$ and $q_\mu$ and which are their main properties.}\label{fig:relations}
\end{figure}

To address the case of $q(G)$  we make use of the  Lov\'asz extension of the modularity set function. The Lov\'asz extension, also referred to as Choquet integral, allows the extension of set valued functions to the entire space $\RR^n$ and is particularly well-suited to deal with optimization of sub-modular functions. We refer to \cite{Bach2011} for a careful introduction to the topic. Below we recall one possible definition of the Lov\'asz	 extension
\begin{definition}
	Given the set of vertices $V$, let $\mathcal P(V)$ be the power set of $V$, and consider a  function $F:\mathcal P(V) \to \RR$. For a given vector $x \in \RR^n$ let $\sigma$ be any permutation such that $x_{\sigma(1)}\leq x_{\sigma(2)}\leq \dots \leq x_{\sigma(n)}$ and let $C_i(x) \subseteq V$ be the set 
	$$C_i(x)=\{k \in V : x_{\sigma(k)}\geq x_{\sigma(i)}\}$$ 
	The Lov\'asz extension $f_F:\RR^n \to \RR$ of $F$ is defined by
	$$f_F(x) = \sum_{i=1}^{n-1} F(C_{i+1}(x))(x_{\sigma(i+1)}-x_{\sigma(i)})+F(V)x_{\sigma(1)}$$
\end{definition}
%
We collect in the next proposition  some useful properties of the Lov\'asz extension, which will be helpful in the following. We refer to \cite{Bach2011} for their proofs.
\begin{proposition}[Some properties of the Lov\'asz extension]\label{prop:Lovasz}
	Consider two set valued functions $F,H:\mathcal P(V) \to \RR$ such that $F(\emptyset)=H(\emptyset)= 0$. Then
	\begin{enumerate}
		\item $f_F+f_H$ is the Lov\'asz extension of $F+H$, i.e. $f_F+f_H=f_{F+H}$.
		\item For all $A\subseteq V$ it holds $F(A)=f_F(\uno_A)$.
		\item $f_F$ is positively one-homogeneous, i.e. $f_F(\alpha x) = \alpha f_F(x)$ for all $\alpha\geq 0$.
		\item Given a graph $G=(V,E)$ let $w:E\to\RR_+$ denote its edge weight function and let $\cut_G$ denote the set valued function $\cut_G(A)=w(E(A,\bar A))$. The Lov\'asz extension of $\cut_G$ is the weighted total variation
		$$
			f_{\cut_G}(x) = \frac 1 2 \sum_{i,j=1}^n w(ij) |x_i-x_j| = \frac 1 2\, |x|_{TV}^w\, .
		$$
		\item Given $x\in \RR^n$ and $t>0$ consider the level set $A_x^t = \{i\in V : x_i>t\}$. Then 
		$$
			f_F(x) = \int_{-\infty}^0 \{F(A_x^t) - F(V)\}\, dt + \int_0^{+\infty}F(A_x^t)\, dt \, .
		$$
	\end{enumerate}  
\end{proposition}
The formula at point $5$ is actually one of the many equivalent definitions of the Lov\'asz extension and is sometimes referred to as  \textit{the co-area theorem}.
\begin{remark}\label{rem:fQ}
 From the proposition above we deduce that $\<x,\M(x)\>_\mu$ is the Lov\'asz extension of the modularity function $Q$ and it corresponds to the difference of two weighted total variations of $x\in \RR^n$. 

 In fact, given a graph $G=(V,E)$ with weight function $w$, consider the complete graph $K_0=(V,V\times V)$  with weight function $w_0(ij)=d_id_j/\vol_G(V)$, where  $d_i = w(E(\{i\},V))$ and $\vol_G(V)=w(E)$ are the  degree of node $i$ and the volume of $G$, respectively. Then, for any $A\subseteq V$ we have 
 $$
 w(E(A,V)) = \sum_{i\in A}d_i = \sum_{i\in A}d_i\sum_{j\in V} d_j/\vol_G(V)=w_0(E(A,V))\, 
 $$
 Therefore, from \eqref{eq:Q(A)} and the identity $w(E(A,\bar A))=w(E(A,V))-w(E(A))$, we can decompose the modularity of a set $A$ into $Q(A) = \cut_{K_0}(A)-\cut_G(A)$,  
where $\cut_G$ is the set valued function defined at point 4 of Proposition \ref{prop:Lovasz}. Combining points 1 and 4 of Proposition~\ref{prop:Lovasz}
we obtain
\begin{align}\label{eq:difference-of-cuts}
	\begin{aligned}
		f_Q(x)&=f_{\{\cut_{K_0}-\cut_G\}}(x)=f_{\cut_{K_0}}(x)-f_{\cut_G}(x)\\
		& =\frac 1 2 \big\{|x|_{TV}^{w_0}-|x|_{TV}^w\big\} 
	=\<x, \M(x)\>_\mu \, .
	\end{aligned}
\end{align}
\end{remark}

The following technical lemma will be useful in the proof  of Theorem \ref{thm:exact-1-modularity-*} below, being one of our two main theorems of the section.
\begin{lemma}\label{lem:rM*}
 Let $F,H:\mathcal P(V) \to \RR$ be set valued functions  such that $0<H(A) \leq 1$ for all $A \subseteq V$ s.t.\ $A\notin\{\emptyset,V\}$. If $F(V)=0$, then 
	$$\max_{A \subseteq V}\frac{F(A)}{H(A)}\geq \frac 1 2 \max_{ \|x\|_\infty\leq 1} f_F(x)\, .$$
\end{lemma}
\begin{proof}
Suppose w.l.o.g.\ that the entries of $x\in \RR^n$ are labeled in ascending order,  that is $x_1\leq \dots\leq x_n$. We have 
 \begin{align*}
		f_F(x)&=\sum_{i=1}^{n-1}F(C_{i+1}(x))(x_{i+1}-x_i)\leq \sum_{i=1}^{n-1}\frac{F(C_{i+1}(x))}{H(C_{i+1}(x))}H(C_{i+1}(x))(x_{i+1}-x_i)
	\end{align*}
	As $0<H(C_{i+1}(x))\leq 1$ and $(x_{i+1}-x_i)\geq 0$ we get
	\begin{align*}
		f_F(x)&\leq \max_{i=2,\dots, n}\frac{F(C_{i}(x))}{H(C_{i}(x))} (x_n-x_1)\leq  \left(\max_{i=1,\dots, n}\frac{F(C_{i}(x))}{H(C_{i}(x))}\right)2\|x\|_\infty
	\end{align*}
	We get as a consequence
	$$\max_{\|x\|_\infty \leq 1} f_F(x)\leq 2\, \max_{\|x\|_\infty \leq 1}\max_{i=1,\dots, n}\frac{F(C_{i}(x))}{H(C_{i}(x))}=2\, \max_{A\subseteq V}\frac{F(A)}{H(A)}$$
	and this proves the claim.
\end{proof}

The above lemma allows us to show that $r_\M^*$ is an exact nonlinear relaxation of the modularity function $q$

\begin{theorem}\label{thm:exact-1-modularity-*}Let $r_\M^*$ be the Rayleigh quotient defined in \eqref{eq:Rayleigh_quotients} and let  $\lambda_1(r_\M^*)=\max_{x\in \RR^n}r_\M^*(x)$. Then $r_\M^*(\uno_A-\uno_{\bar A})= q(A)\mu(V)$, for any $A\subseteq V$ and
	$$ q(G) = \max_{A\subseteq V}q(A) =
	{\lambda_1(r_\M^*)}/{\mu(V)}\, .$$
\end{theorem}
\begin{proof}
For a subset $A \subseteq V$, consider the vector $v_A=\uno_A-\uno_{\bar A}$. Then 
 $$\<v_A,\M(v_A)\>_\mu = \frac 1 2 \sum_{i,j=1}^n \mu_i (-M)_{ij}|(v_A)_i-(v_A)_j|=2 \, Q(A)$$ 
and $\|v_A\|_\infty=1$. Therefore $r_\M^*(v_A)=2 \, Q(A)= q(A)\mu(V)$ and
\begin{equation}\label{eq:x}
 \mu(V)q(G) =  \max_{A\subseteq V}r_\M^*(v_A)\leq \max_{x\in \RR^n}r_\M^*(x)\, .
\end{equation}

To show the reverse inequality we  use Lemma \ref{lem:rM*} and Remark \ref{rem:fQ}. 
By \eqref{eq:difference-of-cuts} we have $f_Q(x)=\<x,\M(x)\>_\mu$. 
Now let $H:\mathcal P(V)\to \RR$ be the constant function $H(A) =1$. As $Q(V)=0$, we can use such $H$ into Lemma \ref{lem:rM*}, with $F=Q$,  to get
$$\max_{A\subseteq V}Q(A) \geq \frac 1 2 \max_{\|x\|_{\infty}\leq 1}\<x,\M(x)\>_\mu = \frac 12 \max_{x \in \RR^n}r_\M^*(x)\, .$$
where the second identity holds since $f_Q$ is positively one-homogeneous (Proposition \ref{prop:Lovasz}, point 3). Combining the latter inequality with \eqref{eq:x}  we conclude. 
\end{proof}

We now prove an analogous result involving $q_\mu(G)$ and $r_\M$,
%
To this end we formulate the following Lemma \ref{lem:f=sup}. The proof is a straightforward modification of the proof of Lemma 3.1 in \cite{hein2011beyond}, and is omitted for brevity.

\begin{lemma}\label{lem:f=sup}
	A function $f:\RR^n\to \RR$ is positively one-homogeneous, even, convex and $f(x+y)=f(x)$ for any $y\in \sp{(\uno)}$ if and only if there exists $\mu:V\to \RR_+$ such that $f(x)=\sup_{y\in Y}\<x,y\>_\mu$ where $Y$ is a closed symmetric convex set such that $\<y,\uno\>_\mu=0$ for any $y\in Y$.
\end{lemma}

The following theorem shows that $r_\M$ is an exact nonlinear relaxation of the normalized modularity function $q_\mu$.

\begin{theorem}\label{thm:exact-1-modularity}Let $r_\M$ be the Rayleigh quotient defined in \eqref{eq:Rayleigh_quotients} and let $\lambda_1^\bot(r_\M)=\max_{x\in \RR^n, \<x,\uno\>_\mu=0}r_\M(x)$. Then $r_\M(\uno_A-\mu(A)/\mu(V)\uno)=q_\mu(A)/2$ for any $A\subseteq V$ and 
$$ q_\mu(G) = \max_{A\subseteq V}q_\mu(A) =2\, 
{\lambda_1^\bot(r_\M)}\,.$$
\end{theorem}

\begin{proof}
For $A\subseteq V$ let $\nu(A)=\mu(A)\mu(\bar A)/\mu(V)$. Then $q_\mu(A)=Q(A)/\nu(A)$. Moreover, if  $w_A = \uno_A - \mu(A)/\mu(V)\uno$, we have  
$\|w_A\|_{1,\mu} = 2\, \nu(A)$ and $\<w_A,\M(w_A)\>_\mu=\<\uno_A,\M(\uno_A)\>_\mu=Q(A)$. Thus  $r_\M(w_A)=q_\mu(A)/2$ and 
\begin{equation}\label{eq:mop}
q_\mu(G)=2\max_{x\in \{-a,b\}^n, \<x,\uno\>_\mu=0}r_\M(x)\, .
\end{equation}

 Now, for $x\in \RR^n$ and $t>0$ consider the level set $A_x^t = \{i\in V : x_i>t\}$ and let  $x_{\min} = \min_i x_i$ and $x_{\max}=\max_i x_i$.  From the co-area formula (Proposition \ref{prop:Lovasz} point 5) and the identity $f_Q(x) = \<x,\M(x)\>_\mu$ shown in \eqref{eq:difference-of-cuts} we have 
$$
\<x,\M(x)\>_\mu = \int_{-\infty}^{+\infty} Q(A_x^t) \, dt =  \int_{x_{\min}}^{x_{\max}}Q(A_x^t)\, dt\, .
$$
Given $A\subseteq V$, let $w_A$ denote the vector $w_A = \uno_A - \mu(A)/\mu(V)\uno$. From $\|w_A\|_{1,\mu}=2\nu(A)$  we obtain
 \begin{align*}
  \<x,\M(x)\>_\mu   \leq \left\{\max_t \frac {Q(A_x^t)}{2 \nu(A_x^t)}\right\}\int_{x_{\min}}^{x_{\max}} 2\nu(A_x^t)\, dt =   \left\{\max_t \frac {Q(A_x^t)}{2 \nu(A_x^t)}\right\}\int_{x_{\min}}^{x_{\max}} \|w_{A_x^t}\|_{1,\mu}\, dt .
 \end{align*}
Let $P:\RR^n\to \RR^n$ be the orthogonal projection onto $\{x:\<x,\uno\>_\mu=0\}$, that is $P(x)=x-\<x,\uno\>_\mu/\mu(V)\uno$, and consider the function $f(x) = \|P(x)\|_{1,\mu}$. Note that $f$ satisfies all the hypothesis of Lemma \ref{lem:f=sup} above. Moreover note that $f(\uno_A)=\|w_A\|_{1,\mu}$ for any $A\subseteq V$. Thus, by Lemma \ref{lem:f=sup}, there exists $Y\subseteq \mathrm{range}(P)$ such that 
$$
\int_{x_{\min}}^{x_{\max}} \|w_{A_x^t}\|_{1,\mu}dt = \sup_{y \in Y}\int_{x_{\min}}^{x_{\max}} \<\uno_{A_x^t},y\>_\mu dt  \, .
$$
Assume w.l.o.g.\ that $x$ is ordered so that $x_1\leq \cdots\leq x_n$. Note that the function $\phi(t)=\<\uno_{A_x^t}, y\>$ is constant on the intervals $[x_i,x_{i+1}]$. Thus, letting $A_i=A_x^{x_i}$ we have
$$
\int_{x_{\min}}^{x_{\max}}\phi(t)dt = \sum_{i=1}^{n-1}(x_{i+1}-x_i)\<\uno_{A_i},y\>_\mu=\sum_{i=1}^n x_i\<\uno_{A_{i-1}}-\uno_{A_i},y\>_\mu = \<x,y\>_\mu \, ,
$$
thus, by Lemma \ref{lem:f=sup}, 
$$
\|P(x)\|_{1,\mu}=f(x) =  \sup_{y\in Y}\<x,y\>_\mu = \int_{x_{\min}}^{x_{\max}} \|w_{A_x^t}\|_{1,\mu}\, dt\, .
$$
 Denote by $A_x^*$  the set that attains the maximum $\max_t Q(A_x^t)/\nu(A_x^t)$. As $\<x,\M(x)\>_\mu=\<P(x),\M(P(x))\>_\mu$, all together we have 
\begin{equation}\label{eq:7}
 \lambda_1^\bot(r_\M)=\max_{x\in \RR^n}r_\M(P(x)) = \max_{x\in \RR^n} \frac{\<x,\M(x)\>}{\|P(x)\|_{1,\mu}}\leq \max_{x\in \RR^n}\frac{Q(A_x^*)}{2\nu(A_x^*)}\leq q_\mu(G)/2\, .
\end{equation}
On the other hand, using \eqref{eq:mop} we get $q_\mu(G)\leq 2 \lambda_1^\bot(r_\M)$ and  
together with \eqref{eq:7} this proves the statement. 
\end{proof}

\section{Spectral method for nonlinear modularity}\label{sec:nonlinear_spectral_method}

As in the spectral method proposed by Newman \cite{newman2006finding}, we can identify a leading module  in the network by partitioning the vertex set into two subsets associated to the maximizers of either  $\lambda_1(r_\M^*)$ or $\lambda_1^\bot(r_\M)$. The pseudo code for the method for $r_\M^*$ is presented below, obvious changes are needed when $r_\M^*$ is replaced by $r_\M$.
\begin{flalign}\label{eq:nonlinear_spectral_method}\tag{M1}
 &\begin{array}{l}
   \begin{cases}
\mathit 1.\, \, \text{Compute }\lambda_1(r_\M^*) \text{ and an associated eigenvector }x\\ 
\mathit 2.\, \, 	\text{If }\lambda_1(r_\M^*)>0:\\
\quad\,   \text{partition the vertex set into }A_+ \text{ and }\bar{A_+}\text{ by optimal thresholding} \\
\quad\,  \text{the eigenvector }x \text{ with respect to the community measure}
\end{cases}\!\!\!\!
  \end{array}&&
\end{flalign}

The optimal thresholding technique for $x$ at step $2$ returns the partition $\{A_+,\bar{A_+}\}$ defined by $A_+=\{i\in V:x_i>t^*\}$, being $t^*$ such that $t^* = \arg\max_{t} q(\{i:x_i>t\})$. 

The procedure \eqref{eq:nonlinear_spectral_method} can be iterated into a successive bi-partitioning strategy which can be sketched as follows:  Consider the nonlinear modularity operator $\M_i$, $i=1,2$, associated with the two subgraphs $G_1=G(A_+)$ and $G_2=G(\bar A_+)$, respectively, and look for a maximal module within $G_1$ and $G_2$  by repeating points \emph{1} and \emph{2}, and so forth. As in the linear case,  each time this procedure is iterated, we have to consider a new  nonlinear modularity operator. If $A\subseteq V$ is the subset of nodes associated with the current recursion, that is $G_i=G(A)$,  
the new nonlinear modularity operator $\M_i$ is defined by replacing the modularity matrix $M$ in  \eqref{eq:nonlinearM}  with the  modularity matrix $M_A$ of the corresponding subgraph $G(A)$, given by \cite{newman2006modularity,fasino2016generalized}
$$(M_A)_{ij} =\left\{ \begin{array}{l}
			\frac 1 {\mu_i}M_{ij} \quad  \text{if } i\neq j\\
                \frac 1 {\mu_i}\big(M_{ii}- (W_{G(A)}\uno)_i + \frac{\vol(A)}{\vol(V)}(W_G\uno)_i\big)
               \end{array}\right. \qquad \text{for } i,j \in A\, ,
$$
where $W_G$ and $W_{G(A)}$ are the weight matrices of the graphs $G$ and $G(A)$, respectively. 

We discuss in what follows a generalized version of the RatioDCA method \cite{hein2011beyond} for approaching step $1$ in the above procedure \eqref{eq:nonlinear_spectral_method}. 
The method converges to a critical value of the Rayleigh quotients \eqref{eq:Rayleigh_quotients} and ensures a better approximation of $q(G)$ and $q_\mu(G)$ than the standard linear spectral method. 

\subsection{Generalized RatioDCA method} 
The RatioDCA technique \cite{hein2011beyond} is a general scheme for minimizing the ratio of nonnegative  differences of convex one-homogeneous functions. We extend that technique to the case where the difference of functions in the numerator can attain both positive and negative values.  As our goal is to maximize $r_\M$ and $r_\M^*$, we then apply the method to $-r_\M$ and $-r_\M^*$ respectively. 

The generalized RatioDCA technique we propose is of self-interest. For this reason, we formulate and analyze the method for general ratio of differences of convex one-homogeneous functions $f_1,f_2,g_1,g_2:\RR^n\to \RR$, such that $g_1(x)-g_2(x)\geq 0$ for all $x\in \RR^n$. 

Define the function
\begin{equation}\label{eq:r(x)}
 r(x) = \frac{f_1(x)-f_2(x)}{g_1(x)-g_2(x)}
\end{equation}
and consider the problem of computing the minimum $\min_x r(x)$. 
The function \eqref{eq:r(x)} can be seen as a generalized Rayleigh quotient and the critical values $\lambda$ of $r(x)$ satisfy the generalized eigenvalue equation
\begin{equation}\label{eq:gen_eigenvalue_r}
0 \in \partial f_1(x)-\partial f_2(x) -\lambda(\partial g_1(x)-\partial g_2(x))\, .
\end{equation}
 In analogy with Definition \ref{def:eig}, when \eqref{eq:gen_eigenvalue_r} holds we say that $\lambda$ is a nonlinear eigenvalue associate to $r$, with corresponding nonlinear eigenvector $x$. Computing the minimum of $r(x)$ is in general a  non-smooth and non-convex optimization problem, so an exact computation of the global minimum of $r(x)$ for general functions and large values of $n$ is out of reach. 
 However, in Theorems~\ref{thm:monotony} and \ref{thm:convergence} we prove that the generalized RatioDCA technique described in Algorithm \ref{alg:genDCA} generates a monotonically descending sequence converging to a nonlinear eigenvalue of $r(x)$.

\RestyleAlgo{ruled}
\begin{algorithm}[t]
	\DontPrintSemicolon
	\caption{Generalized RatioDCA}\label{alg:genDCA}
	\KwIn{Initial guess $x_0$, with $\|x_0\|=1$ and $\lambda_0 = r(x_0)$}
	\Repeat{$|\lambda_{k+1}-\lambda_k|/|\lambda_k|<\text{tolerance}$}{
	\uIf{$\lambda_k\geq 0$}{
		$F_2(x_k)\in \partial f_2(x_k)$, $\quad$ $G_1(x_k)\in \partial g_1(x_k)$\;
		$x_{k+1}=\arg\min_{\|\xi\|_2\leq 1}\Big\{ f_1(\xi)-\<\xi,F_2(x_k)\> + \lambda_k\big(g_2(\xi)-\<\xi,G_1(x_k)\>\big)   \Big\}$\;
	 }
	 \Else{
		$F_2(x_k)\in \partial f_2(x_k)$, $\quad$ $G_2(x_k)\in \partial g_2(x_k)$\;
		$x_{k+1}=\arg\min_{\|\xi\|_2\leq 1}\Big\{ g_1(\xi)-\<\xi,G_2(x_k)\> + \frac 1 {\lambda_k}\big(\<\xi,F_2(x_k)\>-f_1(\xi)\big)   \Big\}$\;
	 }
	 $\lambda_{k+1}=r(x_{k+1})$\;
	}
	\KwOut{Eigenvalue $\lambda_{k+1}$ and associated eigenvector $x_{k+1}$}
\end{algorithm}
The following theorems describe the convergence properties of the generalized RatioDCA algorithm.
\begin{theorem}\label{thm:monotony} Let $\{\lambda_k\}_k$ be the sequence generated by the generalized RatioDCA. Then either $\lambda_{k+1}<\lambda_k$ or the method terminates and it outputs a nonlinear eigenvalue $\lambda_{k+1}$ of $r$ and a corresponding nonlinear eigenvector $x_{k+1}$.
\end{theorem}
\begin{proof}
 Define $\fun_1$ and $\fun_2$ as in Lines 4 and 7 of Algorithm \ref{alg:genDCA}. Namely, 
 $$\fun_1(\xi) = f_1(\xi)-\<\xi,F_2(x_k)\> + \lambda_k\big(g_2(\xi)-\<\xi,G_1(x_k)\>$$
 and 
 $$
 \textstyle{\fun_2(\xi)=g_1(\xi)-\<\xi,G_2(x_k)\> + \frac 1 {\lambda_k}\big(\<\xi,F_2(x_k)\>-f_1(\xi)\big) \, .}
 $$
 By construction we have $\fun_1(x_k)=\fun_2(x_k)=0$, due to the fact that for any convex one-homogeneous function $f$, and any $F(x)\in\partial f(x)$, it holds $\<x,F(x)\>=f(x)$. Recall moreover that, for any convex one-homogeneous function $f:\RR^n\to\RR$, it holds $f(x)\geq \<x,F(y)\>$, for any $x,y \in\RR^n$ and any $F(y)\in \partial f(y)$ (see e.g. \cite{hiriart2012fundamentals}).
 
 If $\lambda_k\geq 0$, by definition of $x_{k+1}$ we have $\fun_1(x_{k+1})\leq 0$. Two cases are possible: either $\fun_1(x_{k+1})< 0$ or $\fun_1(x_{k+1})= 0$. In the first case  we have
 $$
 f_1(x_{k+1})+\lambda_k\, g_2(x_{k+1}) < \< x_{k+1}, F_2(x_k)  \>+\lambda_k \< x_{k+1}, G_1(x_k)  \>\leq f_2(x_{k+1})+\lambda_k g_1(x_{k+1})\,
 $$
 therefore $f_1(x_{k+1})-f_2(x_{k+1})< \lambda_k (g_1(x_{k+1})-g_2(x_{k+1}))$ that is $\lambda_{k+1}< \lambda_k$.
 Otherwise $\fun_1(x_{k+1})= 0$, thus $\lambda_{k+1}= \lambda_k$ and the method terminates. As $f_1,f_2,g_1,g_2$ are one-homogeneous we deduce that $x_{k+1}=x_k$ is a global minimum of $\fun_1$, thus $0\in \partial \fun_1(x_{k+1})$. This implies $0 \in \partial f_1(x_{k+1})-F_2(x_{k+1}) -\lambda_{k+1}(G_1(x_{k+1})-\partial g_2(x_{k+1}))$, that is $\lambda_{k+1}$ is a nonlinear eigenvalue of $r$ with corresponding nonlinear eigenvector $x_{k+1}$. 
 
 Let us now consider the case $\lambda_k<0$. We have 
 $$\fun_2(x_{k+1})=g_1(x_{k+1})-\<x_{k+1},G_2(x_k)\> + \frac 1 {\lambda_k}\big(\<x_{k+1},F_2(x_k)\>-f_1(x_{k+1})\big) \leq 0\, .$$
 If $\fun_2(x_{k+1})<0$, together with $\lambda_k<0$ and $g_1-g_2\geq 0$ this implies
 $$g_1(x_{k+1})-\frac{1}{\lambda_k}f_1(x_{k+1})< \<x_{k+1},G_2(x_{k})\>-\frac{1}{\lambda_k}\<x_{k+1},F_2(x_k)\>\leq g_2(x_{k+1})-\frac 1 {\lambda_k}f_2(x_{k+1})$$
 therefore $g_1(x_{k+1})-g_2(x_{k+1})< -\frac 1 {\lambda_k}(f_2(x_{k+1})-f_2(x_{k+1}))$, that is $\lambda_{k+1}< \lambda_k$. Again, note that the equality holds only if the optimal value in the inner problem is zero, which implies in turn that the sequence terminates and the point $x_{k+1}=x_{k}$ is a critical value of $\fun_2$, thus $0\in \partial \fun_2(x_{k+1})$. We get
 $$0 \in \partial g_1(x_{k+1}) - G_2(x_{k+1}) - (\partial f_1(x_{k+1})-F_2(x_{k+1}))/\lambda_{k+1}\, .$$
Multiplying the previous equation by $-\lambda_{k+1}\neq 0$  we conclude the proof.
\end{proof}

\begin{theorem}\label{thm:convergence}
 Let $\{\lambda_k\}_k\subseteq \RR$ and $\{x_k\}_k\subseteq \RR^n$ be the sequences defined by the generalized RatioDCA method. Then \begin{enumerate}
 \item $\lambda_k$ converges to a nonlinear eigenvalue $\lambda$ of $r$, 
 \item there exists a subsequence of $\{x_k\}_k$ converging to a nonlinear eigenvector of $r$ corresponding to $\lambda$ and the same holds for any convergent subsequence of~$\{x_k\}_k$.
 \end{enumerate}
\end{theorem}
\begin{proof}
 The sequence $\{x_k\}_k$ belongs to the compact set $\{x:\|x\|_2\leq 1\}$ thus $\lambda_k=r(x_k)$ is decreasing and bounded, and there exits a convergent subsequence $x_{k_j}$. We deduce that there exists  $\lambda$ such that $\min_{\|x\|_2\leq 1}r(x)\leq\lambda=\lim_k r(x_k)$ and thus, for any convergent subsequence $x_{k_j}$ of $x_k$, we have $\lim_j x_{k_j}=x_*$ with $r(x_*)=\lambda$.  Similarly to the previous proof, define $\fun_1$ and $\fun_2$ as 
 \begin{align*}
  \fun_1(\xi) &= f_1(\xi)-\<\xi,F_2(x_*)\> + \lambda \big(g_2(\xi)-\<\xi,G_1(x_*)\>\\
  \fun_2(\xi)&=g_1(\xi)-\<\xi,G_2(x_*)\> + \frac 1 {\lambda}\big(\<\xi,F_2(x_*)\>-f_1(\xi)\big) \, .
 \end{align*}
  Assume $\lambda<0$. We observe that $\fun_2$ has to be nonnegative. In fact, let $\tilde x = \arg\min_{\|\xi\|\leq 1}\fun_2(\xi)$ and assume that $\fun_2(\tilde x)<0$. Arguing as in the proof of Theorem \ref{thm:monotony}, we get
  $r(\tilde x)>\lambda =r(x_*)$ which is a contradiction, as $\lambda$ is the limit of the sequence $\lambda_k=r(x_k)$. This implies that $x_*$ is a critical point for $\fun_2$, thus $0\in \partial \fun_2(x_*)$, showing that $x_*$ is a nonlinear eigenvector of $r$ with critical value $\lambda$. If $\lambda\geq 0$, an analogous argument applied to $\fun_1$ leads to the same conclusion, thus concluding the proof.
\end{proof}

\subsection{Generalized RatioDCA for modularity Rayleigh quotients}
In order to apply Algorithm \ref{alg:genDCA} to $r_\M^*$ and $r_\M$ recall that, as observed in \eqref{eq:difference-of-cuts}, the quantity $\<x,\M(x)\>_\mu$ is the difference of two 
weighted total variations $\<x,\M(x)\>_\mu = \frac 1 2 \big\{|x|_{TV}^{w_0}-|x|_{TV}^w\big\}$.
As we aim at maximizing the Rayleigh quotients \eqref{eq:Rayleigh_quotients}, we apply the generalized RatioDCA to either $-r_\M^*$ or $-r_\M$. However, for $r_\M$, we are interested in $\lambda_1^\bot(r_\M)$, and thus we want to maximize $r_\M$ over the subspace $\mathrm{range}(P)$, being $P$ the orthogonal projection $P(x)=x-\<x,\uno\>_\mu/\mu(V)\uno$. 
This issue is addressed by applying the generalized RatioDCA to the function
$$
\tilde r_\M(x) = \frac{\<x,\M(x)\>_\mu}{\, \,\, \,  \|P(x)\|_{1,\mu}} \, .
$$
In fact,  due to the definition of $\M$, we have $r_\M(P(x))=\tilde r_\M(x)$. Thus, optimizing $\tilde r_\M$ is equivalent to optimizing $r_\M$ on the subspace $\mathrm{range}(P)$.

Therefore:
\begin{itemize}
\item In order to address $\lambda_1(r_\M^*)$ we apply Algorithm \ref{alg:genDCA} with the choices $f_1(x)=\frac 12 |x|_{TV}^w$, $f_2(x)=\frac 12 |x|_{TV}^{w_0}$,  $g_1(x)=\|x\|_\infty$ and $g_2(x)=0$.
\item In order to address $\lambda_1^\bot(r_\M)$ we apply Algorithm \ref{alg:genDCA} with the choices $f_1(x)=\frac 12 |x|_{TV}^w$, $f_2(x)=\frac 12 |x|_{TV}^{w_0}$,  $g_1(x)=\|P(x)\|_{1,\mu}$ and $g_2(x)=0$.
\end{itemize}

The following Algorithm \ref{alg:genDCA_rM} shows an implementation of Algorithm \ref{alg:genDCA} tailored to the problem of computing $\lambda_1^\bot(r_\M)$. Straightforward changes are required when implementing the method for $\lambda_1(r_\M^*)$. 

\RestyleAlgo{ruled}
\begin{algorithm}[H]
	\DontPrintSemicolon
	\caption{Generalized RatioDCA for $\lambda_1^\bot(r_\M)$}\label{alg:genDCA_rM}
	\KwIn{Initial guess $x_0\neq 0$ such that $\<x_0, \uno\>_\mu = 0$ and $\lambda_0 = r_\M(x_0)$}
	\Repeat{$|\lambda_{k+1}-\lambda_k|/|\lambda_k|<\text{tolerance}$}{
	$\delta_0(x_k)\in \partial \big\{|x_k|_{TV}^{w_0}\big\}$ such that $\<\uno, \delta_0(x_k)\>_\mu=0$, $\phi(x_k) \in \Phi(x_k)$\;
	\uIf{$\lambda_k\leq 0$}{		
		$y_{k+1}=\arg\min_{\|\xi\|_2\leq 1}\Big\{ 
|\xi|_{TV}^w -\<\xi,\delta_0(x_k)-2\, \lambda_k\,  P\big(\phi(x_k)\big)\>  
\Big\}$\;
	 }
	 \Else{
		$y_{k+1}=\arg\min_{\|\xi\|_2\leq 1}\Big\{ 
2\|P(\xi)\|_{1,\mu}-\frac 1 {\lambda_k}\Big(\<\xi,\delta_0(x_k)\>_\mu - |\xi|_{TV}^w\Big)  \Big\}$\;
	 }
	 $x_{k+1}= P(y_{k+1})$\;
	 $\lambda_{k+1}=r_\M(x_{k+1})$\;
	}
	\KwOut{Eigenvalue $\lambda_{k+1}$ and associated eigenvector $x_{k+1}$}
\end{algorithm}

Note that in the algorithm we need 
to select an element $\delta_0(x)$ of the subdifferential of the total variation of $x$, weighted with $w_0$, being also an element of $\mathrm{range}(P)$, i.e. fulfilling the condition $\<\uno, \delta_0(x)\>_\mu=0$. This is always possible, as long as $x$ is not the constant vector. In fact, consider the sign function $\sigma:\RR\to\{-1,0,1\}$ defined by $\sigma(\lambda)=\lambda/|\lambda|$ if $\lambda\neq 0$ and $\sigma(\lambda)=0$ otherwise. One easily realizes that the  vector $y$, with components
$$
y_i = \frac 1 {\mu_i}\sum_{j=1}^n \frac{d_id_j}{\vol(V)}\sigma(x_i-x_j)\, , \qquad i=1,\dots,n\, ,
$$
 belongs to $\partial \big\{|x_k|_{TV}^{w_0}\big\}$ and is such that $\<\uno, y\>_\mu = 0$, that is $y\in \mathrm{range}(P)$.

A number of  optimization strategies can be used to solve the inner convex-optimization problem at steps $4$ and $6$ of Algorithm \ref{alg:genDCA_rM}. Two efficient methods used in \cite{hein2010inverse,hein2011beyond} are FISTA \cite{beck2009fast} and PDHG \cite{chambolle2011first}. Both methods ensure a quadratic convergence rate. Moreover, 
the computational cost of each iteration of both FISTA and PDHG is led by the cost required to perform the two matrix-vector multiplications $Bx$ and $B^Tx$, being $B$ the node-edge transition matrix of the graph $B:\RR^n\to\RR^{|E|}$, entrywise defined by $(Bx)_{(ij)}=w(ij)(x_i-x_j)$. As it is known, $B$ is typically a very sparse matrix.  We use PDHG in the experiments that we  present in the next section.

Let us conclude with some important remarks related with the practical implementation of the generalized RatioDCA technique. 
First, note that an exact solution of the inner problems at steps $4$ and $6$ is not required in order to ensure monotonic ascending. In fact, the proof of Theorem \ref{thm:monotony} goes through unchanged if $x_{k+1}$ is replaced by any vector $y$ such that $\fun_1(y)<\fun_1(x_k)$, resp.\ $\fun_2(y)<\fun_2(x_k)$. Therefore one can speed up the inner problem phase by computing any $y$ with such a property, especially at an early stage, when the solution is far from the limit.

Second,  Theorem \ref{thm:monotony} ensures that the sequence  of approximations of the Rayleigh quotient generated by the generalized RatioDCA scheme is monotonically increasing. As a consequence, if we run the algorithm by using the leading eigenvector of the modularity matrix $M$ as a starting vector $x_0$, the output is guaranteed to be a better approximation of the modularities $q(G)$ and $q_\mu(G)$.  On the other hand, convergence to a global optimum is not ensured, so in practice one runs the method with a number of starting points and chooses the solution having largest modularity. An  effective choice of the starting point can be done by exploiting a diffusion process on the graph, as suggested in \cite{bresson2013multiclass}. We shall discuss this with more detail in Section \ref{sec:starting_points}.

\section{Numerical experiments}\label{sec:experiments} In this section we apply our method to several real-world networks with the aim of highlighting the improvements that the nonlinear modularity ensures over the standard  linear approach. All the experiments shown in what follows assume $\mu=d$, that is each vertex is weighted with its degree. We subdivide the discussion as follows. 
In Section \ref{sec:balance-nobalance} we discuss
the differences between identified communities associated to 
the exact nonlinear relaxations $r_\M^*$ and $r_\M$ of the modularity and normalized modularity set functions, respectively.
%
%
Then, in Sections \ref{sec:49mnist} and \ref{sec:real-networks}, we focus only on the optimization of the modularity function $q$ and compare the proposed nonlinear approach with other standard techniques. Precisely, in Section \ref{sec:49mnist} we analyze the handwritten digits dataset known as MNIST, restricting our attention to the subset made by the digits $\mathit 4$ and $\mathit 9$. 
We show several statistics including modularity value and clustering error. Finally, in Section \ref{sec:real-networks} we perform community detection on several complex networks borrowed from different applications, comparing the modularity value  obtained with the generalized RatioDCA method for $\lambda_1(r_\M^*)$ against standard methods. We also discuss some experiments where multiple communities are computed. 

\subsection{On the difference between $q(G)$ and $q_\mu(G)$: unbalanced community structure}\label{sec:balance-nobalance}
There are many situations where the community structure in a network is not balanced. Communities of relatively small size can be present in a network alongside communities with a much larger amount of nodes. It is in fact not difficult to imagine the situation of a social network of individual relationships made by communities of highly different sizes. However, a known drawback of modularity maximization \cite{fortunato2007resolution,lancichinetti2011limits} is the tendency to overlook  small-size communities, even if such groups are  well interconnected and can be clearly identified as communities. Many possible solutions to this phenomenon have been proposed in the recent literature, as for instance through the introduction of a tunable resolution parameter $\gamma$, by introducing weighted self-loops, or by considering different null-models  (see \cite{reichardt2006statistical,traag2011narrow,fasino2016generalized}, e.g.). In \cite{zhang2013normalized, fasino2016modularity} it is pointed out that the use of a normalized modularity measure $q_\mu$  is a further potential approach. 
In fact, if we seek at localizing a set $A\subseteq V$ with high modularity $Q(A)$ but  relatively small size $\mu(A)$, then we expect the maximum of $q_\mu$ to be a good indicator of the partition involving $A$. 

In this section we compare the community structure obtained from applying the nonlinear spectral method with $r_\M^*$ and with $r_\M$, aiming at maximizing the modularity  $q$ and the normalized modularity $q_\mu$ functions, respectively. In Figure~\ref{fig:clusters} we show the clustering obtained on a synthetic dataset built trying to model the situation considered in Fig.\ 2 of \cite{fortunato2007resolution}: two small communities poorly connected with each other and with the rest of the network. 

Our aim is to localize the small community as the leading module in the graph. In our synthetic model we generate a random graph $G=(V,E)$ as follows: The small community $A_1$ has 50 nodes, each two nodes in $A_1$ are connected with probability $0.6$, and the weight function for $G$ is such that  $w(ij)=2$ for any $ij \in E(A_1)$. Another group $A_2\subseteq V$ has $100$ nodes, each two nodes in $A_2$ are connected with probability $0.4$, and the weight function for $G$ is such that  $w(ij)=1$ for any $ij \in E(A_2)$. Finally, the rest of the graph $V \setminus(A_1\cup A_2)$ consist of $450$ nodes  and each of them is connected by an edge $ij$ with probability $0.05$ and $w(ij)=1$.

The weight matrix of the graph is shown on the left-most side of Fig \ref{fig:clusters}, whereas the table in the right-most part shows the value of the modularities $q(C_i)$ and $q_\mu(C_i)$ evaluated on the three different partitions $\{C_i,\bar {C_i}\}$, $i=1,2,3$, obtained by the linear spectral method, the nonlinear spectral method with $r_\M^*$ and the one for $r_\M$, respectively. Although the modularity obtained applying the nonlinear spectral method to $r_\M^*$ is the highest one, as expected, the clustering shown in Figure \ref{fig:clusters} highlights how the unbalanced solution obtained through  $\lambda_1^\bot(r_\M)$ is able to recognize the small community $A_1$, whereas the other approaches are not.
\begin{figure}[h]
\begin{minipage}{\linewidth}
 \hfill
\begin{minipage}{.16\linewidth}
\centering
\includegraphics[width=0.9\linewidth, trim =  60mm 92mm 43mm 72mm, clip]{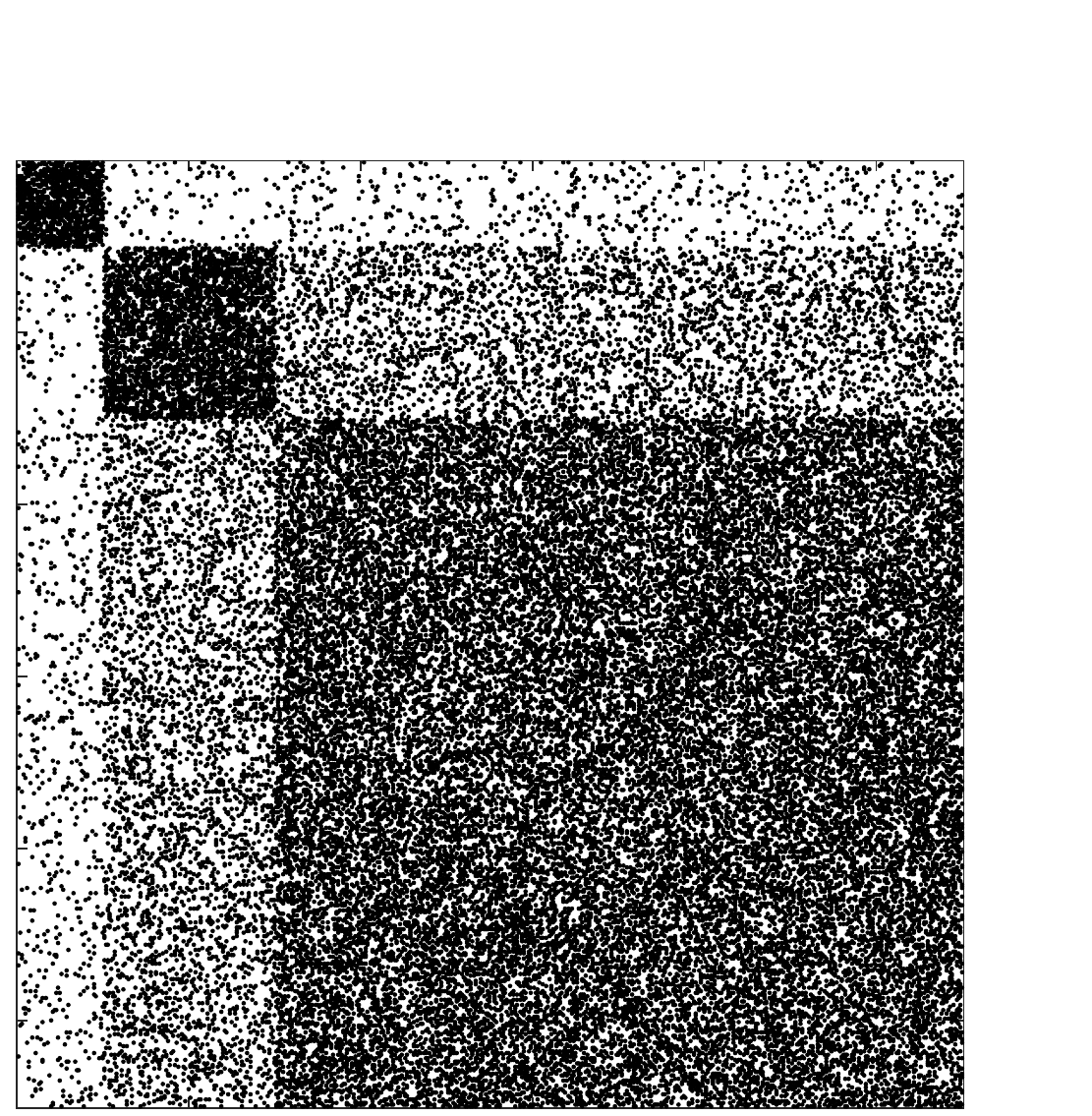}\\
\footnotesize{$W_G$}
\end{minipage}
\hfill
\begin{minipage}{.18\linewidth}
\centering
\includegraphics[width=0.9\textwidth]{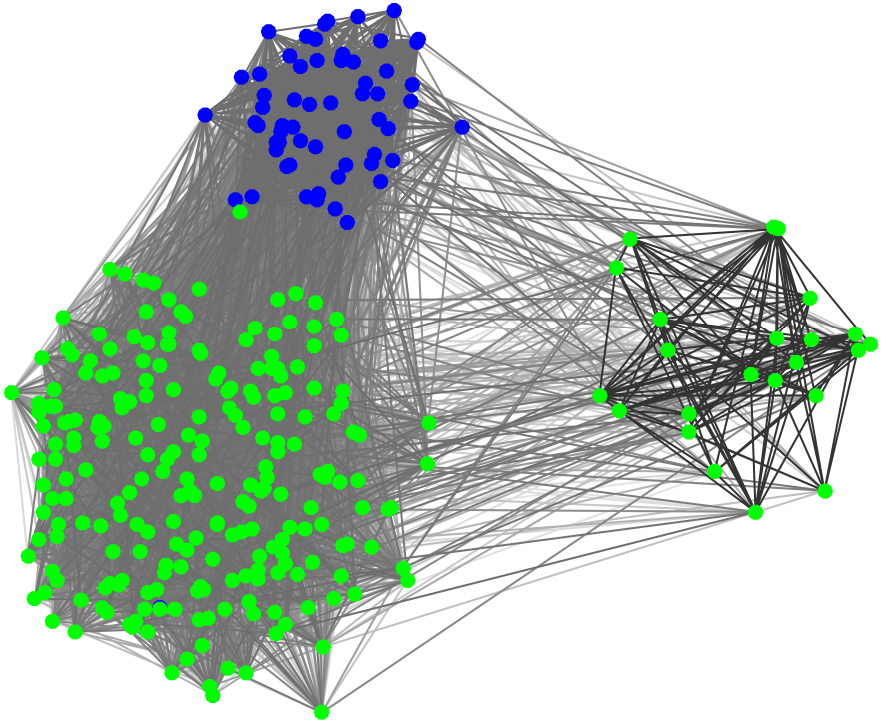}\\
\footnotesize{$(a)$ $\lambda_1(M)$}
\end{minipage}
\hfill
\begin{minipage}{.18\linewidth}
\centering
\includegraphics[width=0.9\textwidth]{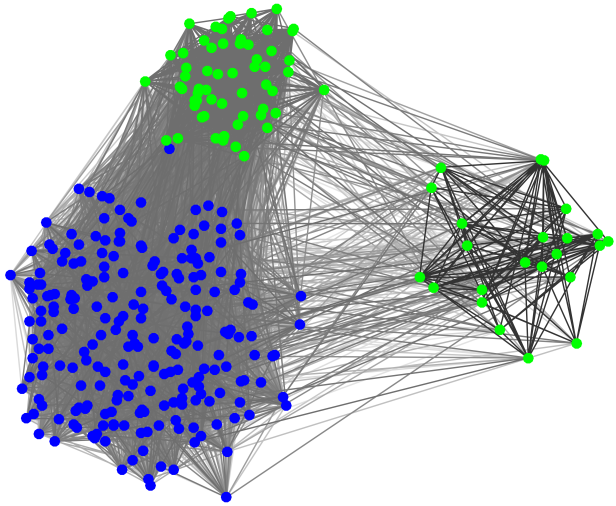}\\
\footnotesize{$(b)$ $\lambda_1(r_\M^*)$}
\end{minipage}
\hfill
\begin{minipage}{.18\linewidth}
\centering
\includegraphics[width=0.9\textwidth]{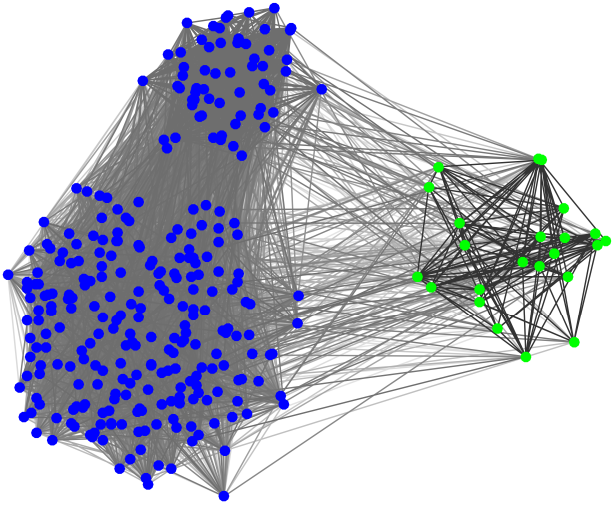}\\
\footnotesize{$(c)$ $\lambda_1^\bot(r_\M)$}
\end{minipage}
\hfill
\begin{minipage}{.19\linewidth}
\begin{tabular}{|c|c|c|}
\hline
 & \scriptsize{\!\!\!$q(C_i)$\!\!\!} & \scriptsize{\!\!\!$q_\mu(C_i)$\!\!\!} \\
 \hline
 \scriptsize{\!\!\!$(a)$\!\!\!} & \scriptsize{\!\!\!0.29\!\!\!} & \scriptsize{0.012\!\!\!}\\
 \scriptsize{\!\!\!$(b)$\!\!\!} & \scriptsize{\!\!\!0.37\!\!\!} & \scriptsize{0.022\!\!\!}\\
 \scriptsize{\!\!\!$(c)$\!\!\!} & \scriptsize{\!\!\!0.13\!\!\!} & \scriptsize{0.029\!\!\!}\\
 \hline
\end{tabular}
\end{minipage}
\hfill
\end{minipage}
\caption{Experiments on synthetic data.
From left to right: Sparsity pattern (spy) plot of the weight matrix of the graph; partition $\{C_1,\bar{C_1}\}$ obtained through Newman's spectral method; partitions $\{C_2,\bar{C_2}\}$ and $\{C_3,\bar{C_3}\}$ obtained through \eqref{eq:nonlinear_spectral_method} with $\lambda_1(r_\M^*)$ and $\lambda_1^\bot(r_\M)$, respectively; value of the modularity of the three partitions.
Relation between matrix spy ($W_G$) and the graph drawings:
the smallest ground-truth community (top-left block of $W_G$) corresponds to the right-most community in graph displays of (a), (b) and (c), whereas the largest community (bottom-right block of $W_G$) is displayed as the bottom-left community in (a), (b) and (c).}\label{fig:clusters}
\end{figure}

In Figure \ref{fig:jazz} we propose a similar comparison made on the Jazz bands network \cite{gleiser2003community}.
 
The network has been obtained from ``The Red Hot Jazz Archive'' digital database, and includes 198 bands that performed between 1912 and 1940, with most of the bands performing in the 1920's. In this case each vertex corresponds to a band, and an edge between two bands is established if they have at least one musician in common. A relatively small community seems to be captured by the normalized modularity $q_\mu$, corresponding to an unbalanced subdivision of the network, whereas a relatively poor community structure corresponds to the standard modularity. The graph drawings are realized by means of the Kamada-Kawai algorithm \cite{kamada1989algorithm}.

\begin{figure}[h]
\begin{minipage}{\linewidth}
 \hfill
\begin{minipage}{.18\linewidth}
\centering
\includegraphics[width=0.99\linewidth, trim =  10mm 60mm 10mm 40mm, clip]{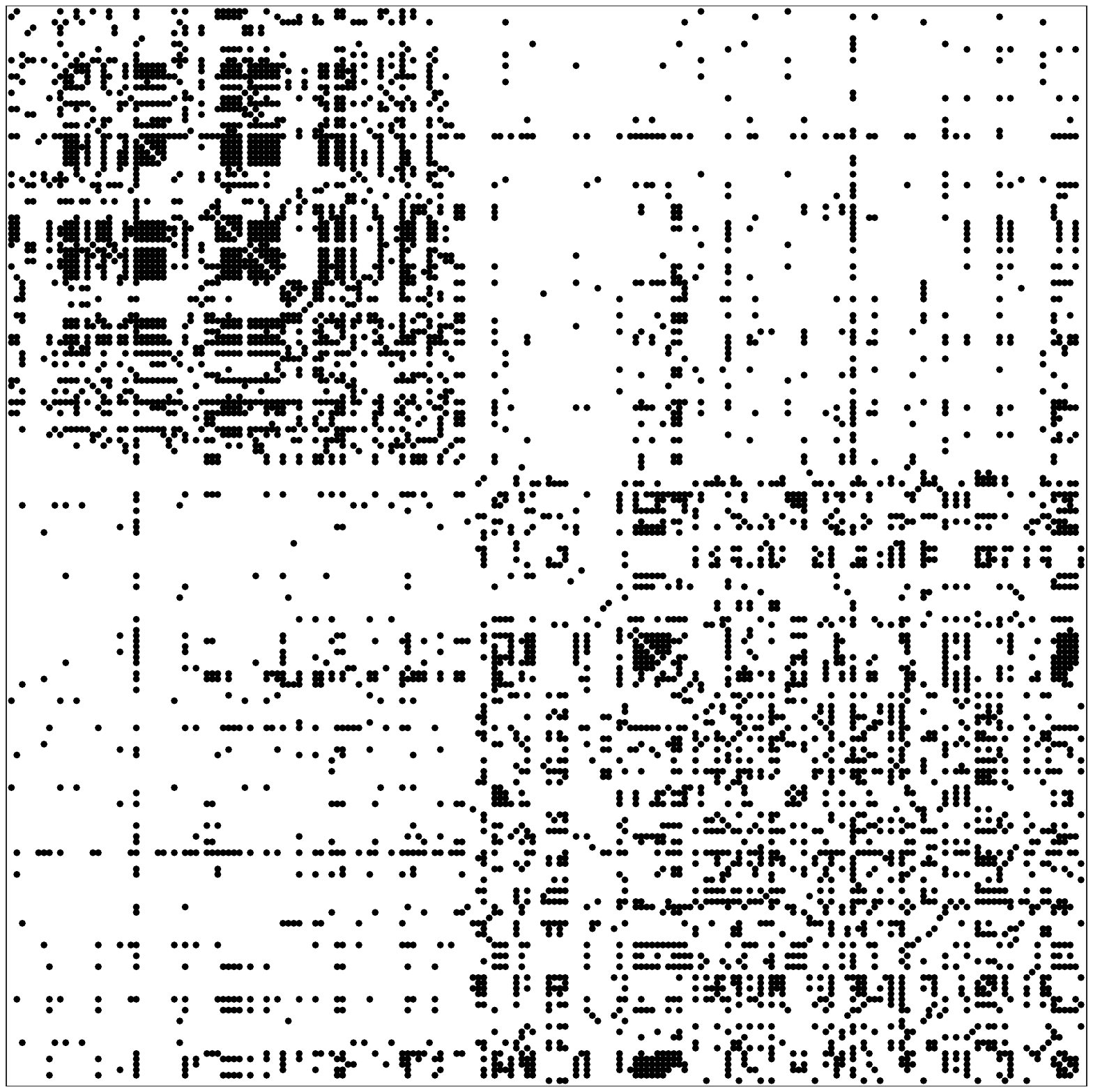}\\
\footnotesize{$W_G$}
\end{minipage}
\hfill
\begin{minipage}{.18\linewidth}
\centering
\includegraphics[width=0.99\textwidth]
{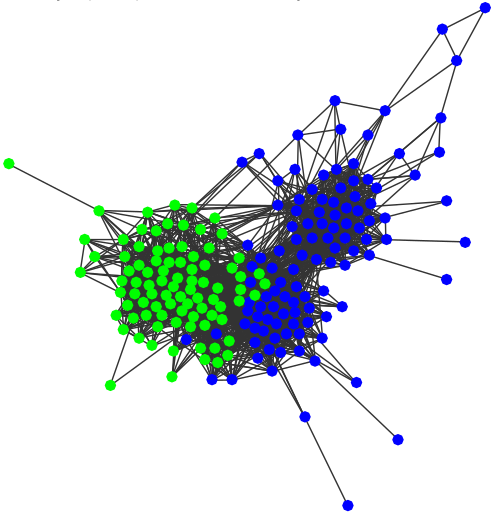}\\
\footnotesize{$(a)$ $\lambda_1(M)$}
\end{minipage}
\hfill
\begin{minipage}{.18\linewidth}
\centering
\includegraphics[width=0.99\textwidth]
{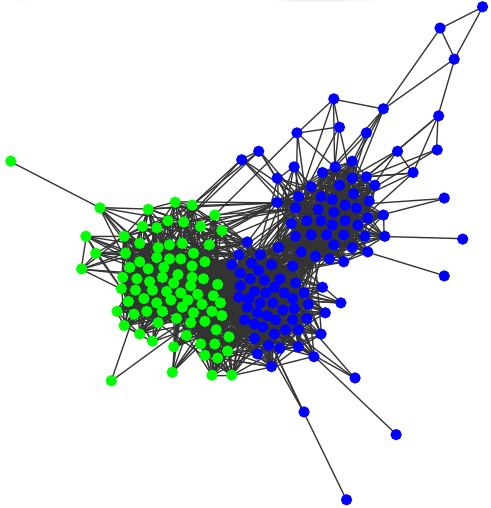}\\
\footnotesize{$(b)$ $\lambda_1(r_\M^*)$}
\end{minipage}
\hfill
\begin{minipage}{.18\linewidth}
\centering
\includegraphics[width=0.99\textwidth]{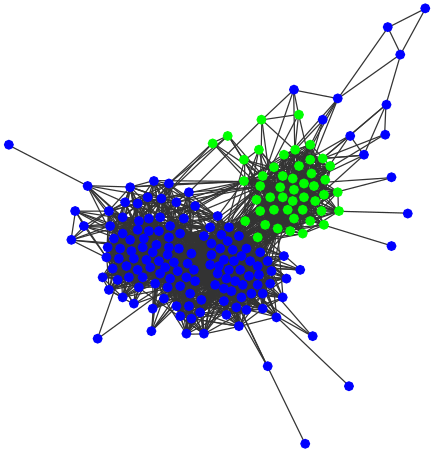}\\
\footnotesize{$(c)$ $\lambda_1^\bot(r_\M)$}
\end{minipage}
\hfill
\begin{minipage}{.19\linewidth}
\begin{tabular}{|c|c|c|}
\hline
 & \scriptsize{\!\!\!$q(C_i)$\!\!\!} & \scriptsize{\!\!\!$q_\mu(C_i)$\!\!\!} \\
 \hline
 \scriptsize{\!\!\!$(a)$\!\!\!} & \scriptsize{\!\!\!0.30\!\!\!} & \scriptsize{0.035\!\!\!}\\
 \scriptsize{\!\!\!$(b)$\!\!\!} & \scriptsize{\!\!\!0.32\!\!\!} & \scriptsize{0.038\!\!\!}\\
 \scriptsize{\!\!\!$(c)$\!\!\!} & \scriptsize{\!\!\!0.27\!\!\!} & \scriptsize{0.050\!\!\!}\\
 \hline
\end{tabular}
\end{minipage}
\hfill
\end{minipage}
\caption{Experiments on Jazz Network. From left to right: Sparsity pattern (spy) plot of the weight matrix of the graph; partition $\{C_1,\bar{C_1}\}$ obtained through Newman's spectral method; partitions $\{C_2,\bar{C_2}\}$ and $\{C_3,\bar{C_3}\}$ obtained through \eqref{eq:nonlinear_spectral_method} with $\lambda_1(r_\M^*)$ and $\lambda_1^\bot(r_\M)$, respectively; value of the modularity of the three partitions. }\label{fig:jazz}
\end{figure}

\subsection{MNIST: handwritten \textit{4-9} digits}\label{sec:49mnist}
The database known as MNIST \cite{lecun1998mnist} consists of 70K images of 10 different handwritten digits ranging from \textit{0} to \textit{9}. This dataset is a widespread benchmark for graph partitioning and data mining. Each digit is an image of $28\times 28$ pixels which is then represented as a real matrix $X_i \in \RR^{28\times 28}$. Here we do not apply any form of dimension reduction strategy, as for instance projection on principal subspaces. For a chosen integer $m$, we build a weighted graph $G=(V,E)$ out of the original data points (images) $X_i$ by placing an edge between node $i$ and its $m$-nearest neighbors $j$, weighted by 
$$w(ij) = \exp\left(-\frac{4\|X_i-X_j\|_F^2}{\min\{\nu(i), \nu(j)\}} \right), \quad \nu(s) = \min_{t:\, st\in E}\|X_s-X_t\|_F^2\, ,$$
being $\|\cdot\|_F$ the Frobenius norm. We limit our attention to the subset of samples representing the digits \textit{4} and \textit{9} which  result into a graph with 13,782 nodes. We refer to this dataset as \textit{49}MNIST.   The reason for choosing such two digits is due to the fact that they are particularly difficult to distinguish, as handwritten \textit{4} and \textit{9}  look very similar (see f.i.\ \cite{hein2011beyond}). 

Although the use of  MNIST dataset is not common in the community detection literature, it gives us a ground-truth community structure to which compare the result of our methods and thus allows for a clustering error measurement. 
In the following Table \ref{tab:49mnist} we compare linear and nonlinear spectral methods on \textit{49}MNIST for different values of $m$ (the number of nearest neighbors defining the edge set of the graph), ranging among $\{5, 10, 15, 20\}$. As the two groups we are looking for are known to be of approximately same size, we apply the nonlinear method \eqref{eq:nonlinear_spectral_method} with $\lambda_1(r_\M^*)$, 
i.e. with the exact nonlinear relaxation of the modularity set function $q$.

Let $\{A,\bar A\}$ be the ground-truth partition of the graph, and let $\{A_+,\bar {A_+}\}$ be the partition obtained by the spectral method. Table \ref{tab:49mnist} shows the following measurements:

\textbf{Modularity.} This is the modularity value $q(A_+)$ of the partition $\{A_+,\bar{A_+}\}$ computed by optimal thresholding the eigenvector of $\lambda_1(M)$ and $\lambda_1(r_\M^*)$, respectively.

\textbf{Clustering error.} This error measure counts the fraction  of incorrectly assigned labels with respect to the ground truth. Namely
$$\text{C.Error} = \frac{1}{n}\Big\{ \sum_{i \in {A_+}}\delta(L_i, L_{A_+}) + \sum_{i \in \bar{A_+}} \delta(L_i,L_{\bar{A_+}}) \Big\}$$
where $\delta$ is the Dirac function, $L_i$ is the true label of node $i$, and $L_{A_+}$, $L_{\bar{A_+}}$ are the dominant true-labels in the clusters $A_+$ and $\bar{A_+}$, respectively. 

\textbf{Normalized Mutual Information (NMI).}  This is an entropy-based similarity measure comparing two partitions of the node set. This  measure is borrowed from information theory, where was originally used to evaluate the Shannon information content of random variables. The Shannon entropy  of a discrete random variable $X$, with distribution $p_X(x)$, is defined by $H(X)=-\sum_x p_X(x)\log p_X(x)$, whereas the mutual information of two discrete random variables $X$ and $Y$~is~defined~as
$$I(X,Y)=\sum_x\sum_y p_{(X,Y)}(x,y)\log\left(\frac{p_{(X,Y)}(x,y)}{p_X(x)p_Y(y)}\right)\, .$$
Finally the NMI of $X$ and $Y$ is $NMI(X,Y)={2 I(X,Y)}/\{H(X)+H(Y)\}$.
The use of NMI for comparing network partitions has been then proposed in  \cite{danon2005comparing,fred2006learning,lancichinetti2009community}. 
\newcommand\VRule[1][\arrayrulewidth]{\vrule width #1}

\begin{table}[!h]
\setlength\extrarowheight{1pt}
\setlength{\tabcolsep}{8pt}
\centering        
\begin{tabular}{!{\VRule[.8pt]}ccccc!{\VRule[.8pt]}}   
 \specialrule{.8pt}{0pt}{0pt} 
            $m$& Method & $q$ & {C.Error} &  {$NMI$}  \\
\specialrule{.8pt}{0pt}{0pt} 
\multirow{2}{*}{5} & \footnotesize{$\lambda_1(M)$} & 0.79 & 0.30 & 0.14  \\   
 & \footnotesize{$\lambda_1(r_\M^*)$}                & 0.95 & 0.01 & 0.88  \\   
\specialrule{.1pt}{0pt}{0pt}                                                                        
\multirow{2}{*}{10} & \footnotesize{$\lambda_1(M)$} & 0.71 & 0.23 & 0.36   \\
 & \footnotesize{$\lambda_1(r_\M^*)$}                 & 0.93  & 0.03  & 0.81    \\  
\specialrule{.1pt}{0pt}{0pt} 
\multirow{2}{*}{15} & \footnotesize{$\lambda_1(M)$} & 0.77 & 0.39 & 0.05   \\   
 & \footnotesize{$\lambda_1(r_\M^*)$}                 & 0.91 & 0.03 & 0.82  \\
\specialrule{.1pt}{0pt}{0pt} 
\multirow{2}{*}{20} & \footnotesize{$\lambda_1(M)$}  & {0.80} & {0.41} & {0.03}   \\
 & \footnotesize{$\lambda_1(r_\M^*)$}      			& {0.91} & {0.03} & {0.82}   \\
 \specialrule{.8pt}{0pt}{0pt} 
\end{tabular}
\caption{Experiments on \textit{49}MNIST dataset and the associated network built out of a $m$-nearest-neighbors graph, with $m \in \{5,10,15,20\}$. With $\lambda_1(M)$ and $\lambda_1(r_\M^*)$ we indicate the linear  method and  our nonlinear variant \eqref{eq:nonlinear_spectral_method}, respectively.}\label{tab:49mnist}
\end{table}

\subsection{Community detection on complex networks}\label{sec:real-networks}
In this section we apply the method \eqref{eq:nonlinear_spectral_method} 
with $r_\M^*$, i.e. the exact nonlinear relaxation of the modularity set function $q$, 
to analyze the community structure of several complex networks of different sizes and representing data taken from different fields, including ecological networks (such as Benguela, Skipwith, StMarks, Ythan2), social and economic networks (such as SawMill, UKFaculty, Corporate, Geom, Erd\H{o}s), protein-protein interaction networks (such as Malaria, Drugs, Hpyroli, Ecoli, PINHuman), technological and informational networks (such as Electronic2, USAir97, Internet97, Internet98,  AS735, Oregon1), transcription networks (such as YeastS), and citation networks (such as AstroPh, CondMat).  Overall we have gathered 68 different networks with sizes ranging from $n=29$ to $n=23133$, all of whom are freely available online. We show the complete list of data sets in Appendix \ref{app:network_names}. 

For each of them we look for the leading module with respect to the unbalanced modularity measure $q$. In particular,  we apply the generalized RatioDCA for $\lambda_1(r^*_\M)$. As this method does not necessarily converge to the global maximum, we run it with different starting points 
and then take as a result the one achieving higher modularity. We discuss the choice of the starting points with more detail in Subsection~\ref{sec:starting_points}. Note that, due to Theorem \ref{thm:monotony}, the choice of the eigenvector  corresponding to $\lambda_1(M)$  as  starting point ensures improvement with respect to the linear case and is often an effective choice. Table \ref{tab:best_2clusters} shows results in this sense: we compare the number of times the nonlinear spectral method outperforms the linear one (in terms of modularity value), with different strategies for the starting point. 
\begin{table}[h]
 \setlength\extrarowheight{1pt}
\centering
{\footnotesize 
\begin{tabular}{!{\VRule[0.8pt]}c!{\VRule[0.8pt]}cccc!{\VRule[0.8pt]}}
\specialrule{0.8pt}{0pt}{0pt}        
Starting point strategy & Eig & 30 Rand & 30 Diff & All\\
\specialrule{.8pt}{0pt}{0pt}                                                                                                                                                                                                      
\textit{Best}			& 100\% & 82.35\% & 95.59\% & 100\%  \\ 
\textit{Strictly Best}   & 95.59\% & 82.35\% & 94.12\% & 97.06\%  \\ 
\specialrule{.8pt}{0pt}{0pt}           
\end{tabular}
}
\caption{
Experiments on real world networks looking for two communities.
Fraction of cases where the nonlinear spectral method \eqref{eq:nonlinear_spectral_method} achieves best and strictly best modularity value $q$ with respect to the linear method. Columns from left to right show results for different sets of starting points: linear modularity eigenvector as starting point, 30 uniformly random starting points, 30 diffused starting points (see Sec.\ \ref{sec:starting_points}), all of them. Experiments are done on 68 networks, listed in  Appendix \ref{app:network_names}.}\label{tab:best_2clusters}
\end{table}

Table~\ref{tab:experimentsOnRealNetworks_2communities}  shows modularity values obtained by the linear spectral method for $\lambda_1(M)$ and our nonlinear spectral technique \eqref{eq:nonlinear_spectral_method}  for $\lambda_1(r_\M^*)$, with the generalized RatioDCA described in Algorithm \ref{alg:genDCA}, on 15 example networks. For the results of our method the shown values are the best value of modularity obtained with the spectral method \eqref{eq:nonlinear_spectral_method} run with  with 61 starting points: 30 random, 30 diffusive (see Sec.\ \ref{sec:starting_points}) and the leading eigenvector of $M$.
The linear modularity approach is outperformed by our nonlinear method:  
the improvement over the modularity matrix linear approach is up to $128\%$, which corresponds to the case of AS735. Also, the size of the modules identified by the two methods often significantly differ.

In Figure \ref{fig:block-plots} we show further statistics on this experiment. In particular, the first plot on the left shows medians and quartiles of the modularity value obtained by the nonlinear method with 61 starting points, highlighting the value obtained with the linear modularity eigenvector as a starting point (magenta triangle) and the best value obtained (black dot). These modularity values are compared with the  modularity values obtained with the linear method (green triangle). The second plot on the right of Figure \ref{fig:block-plots} shows timing performances of the nonlinear method \eqref{eq:nonlinear_spectral_method}  for $\lambda_1(r_\M^*)$ implemented via the generalized RatioDCA  Algorithm \ref{alg:genDCA} on the 15 datasets here considered. 
The generalized RatioDCA is here implemented using PDHG as inner-optimization method \cite{chambolle2011first}.

Finally, in Figure \ref{fig:nobalance} we show  graph drawings comparing the bi-partitions obtained with the two methods on some sample networks. We consider this drawing give a good qualitative intuition of the advantages obtained by using our nonlinear method.

\begin{figure}[h!]
 \begin{minipage}{.13\textwidth}
\centering
\includegraphics[width=0.9\textwidth]{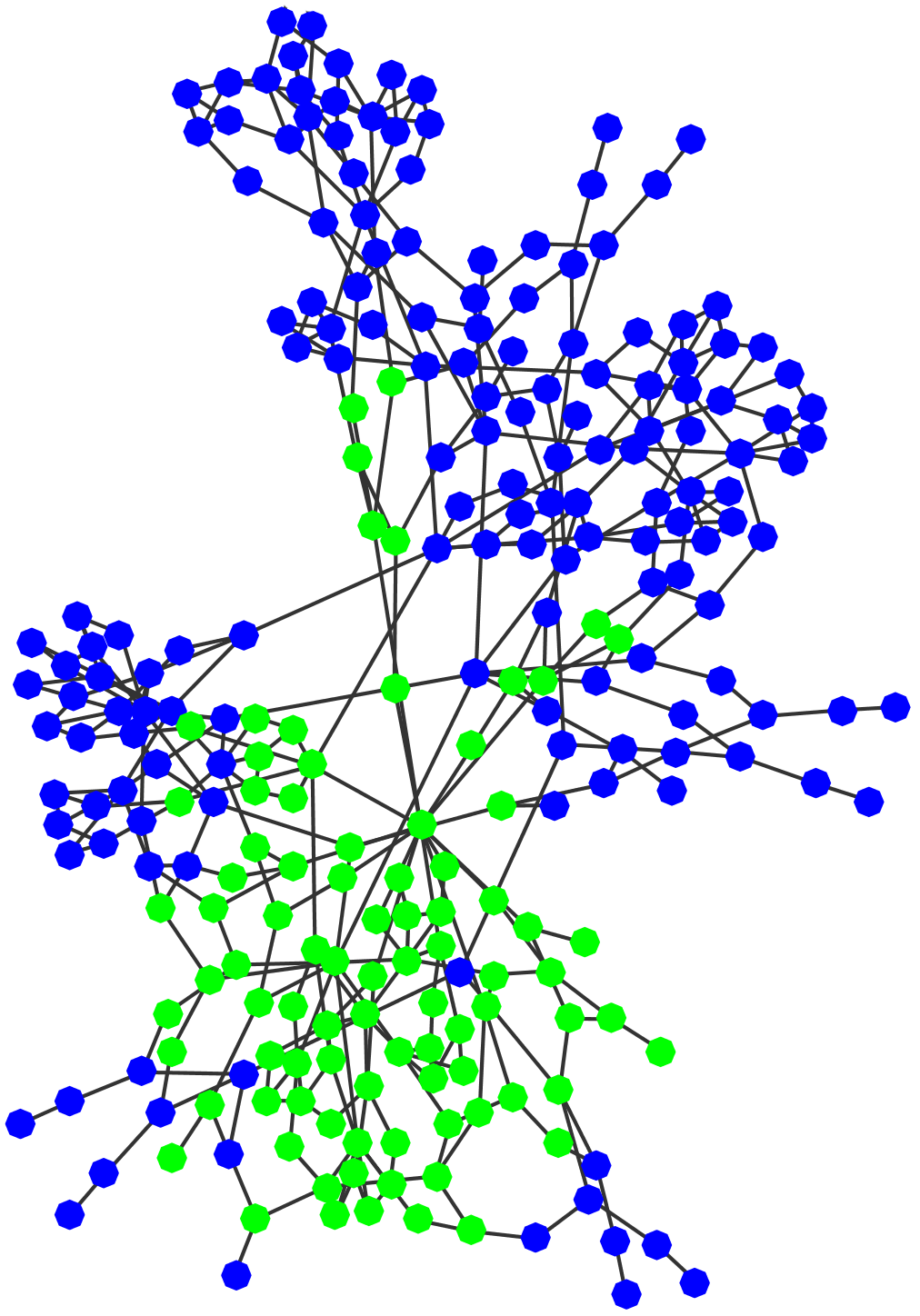}
 \end{minipage}
 \begin{minipage}{.13\textwidth}
\centering
\includegraphics[width=0.9\textwidth]{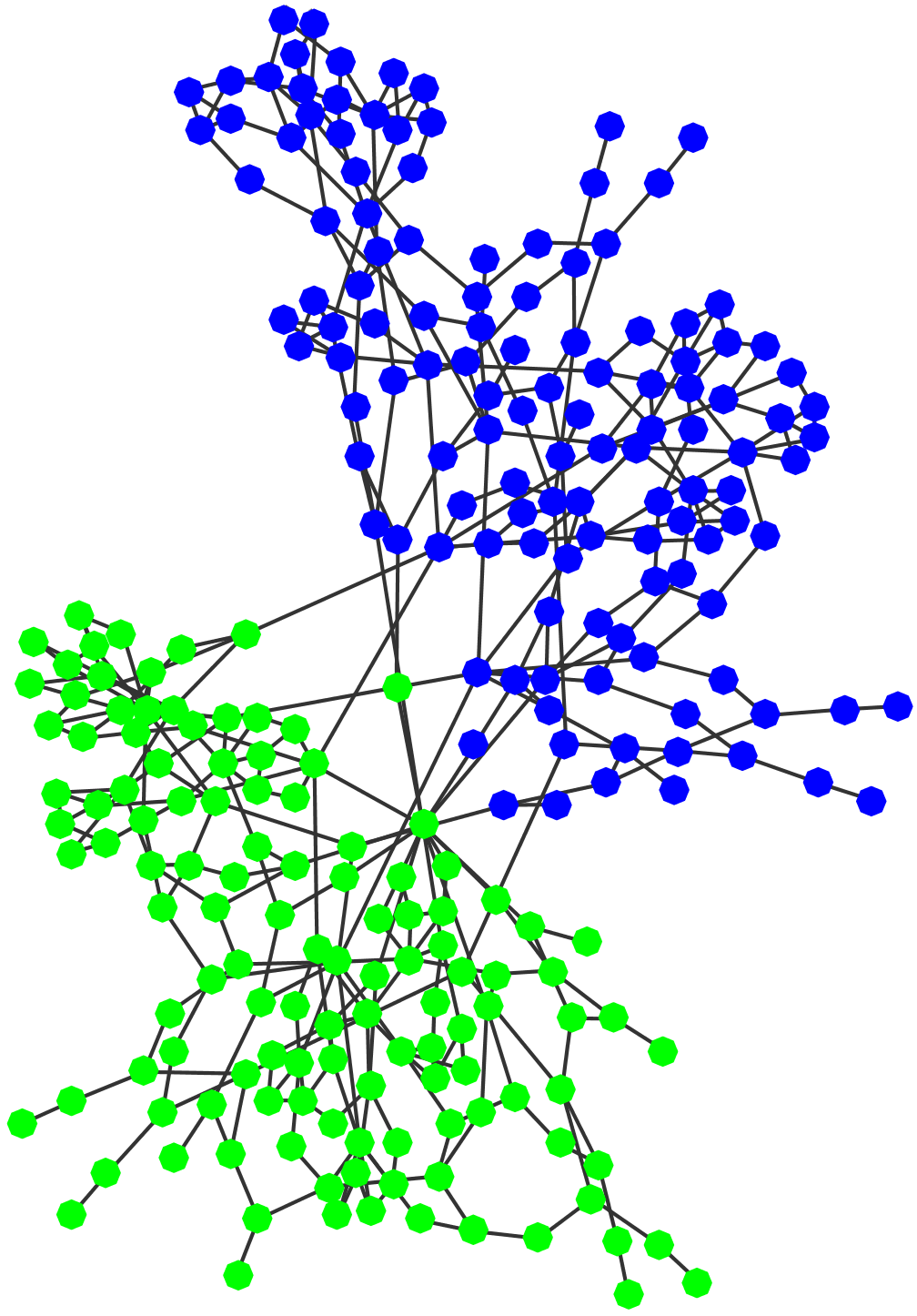}
 \end{minipage}
 \hfill
  \begin{minipage}{.15\textwidth}
\centering
\includegraphics[width=0.9\textwidth]{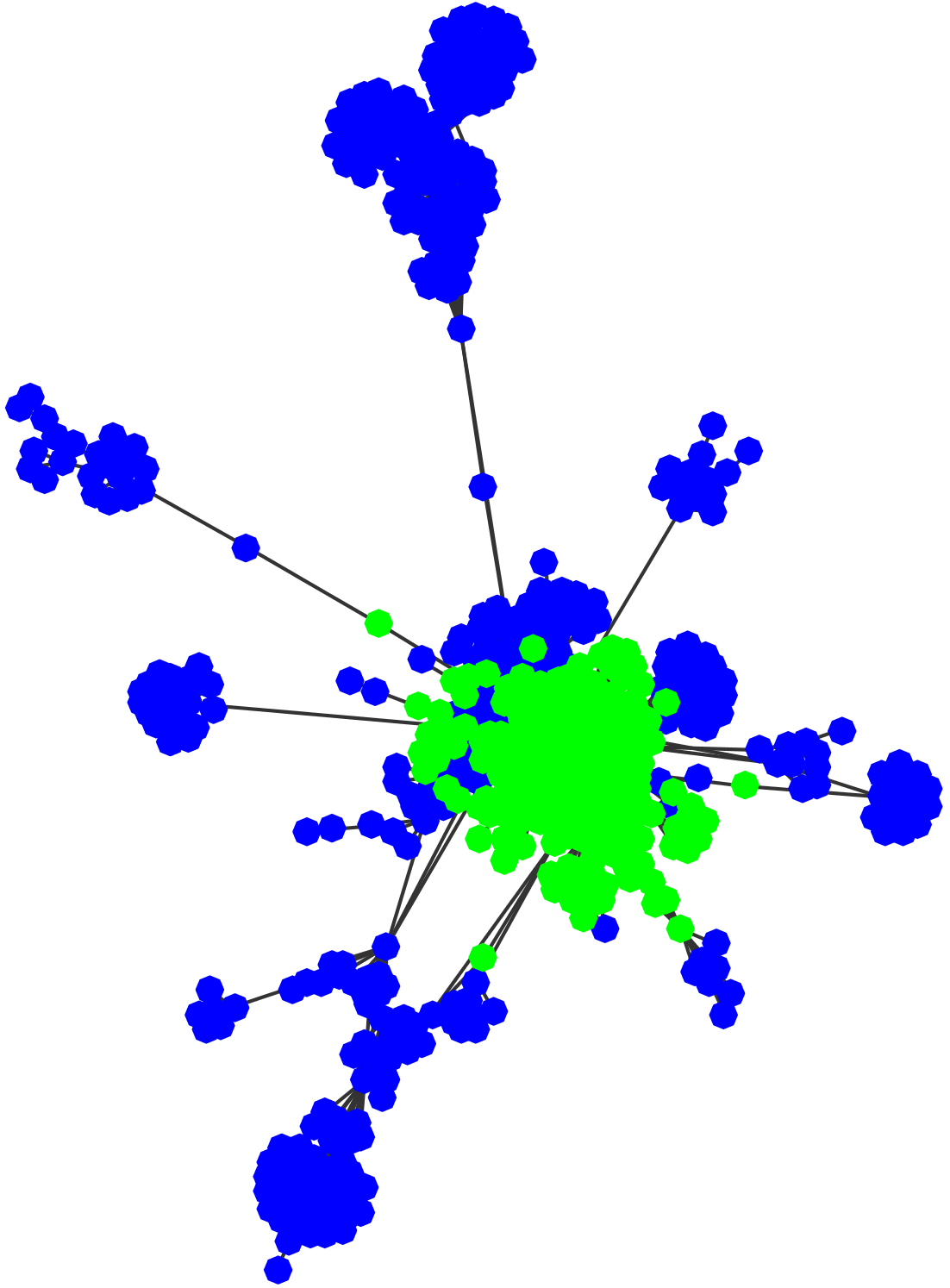}
 \end{minipage}
 \begin{minipage}{.15\textwidth}
\centering
\includegraphics[width=0.9\textwidth]{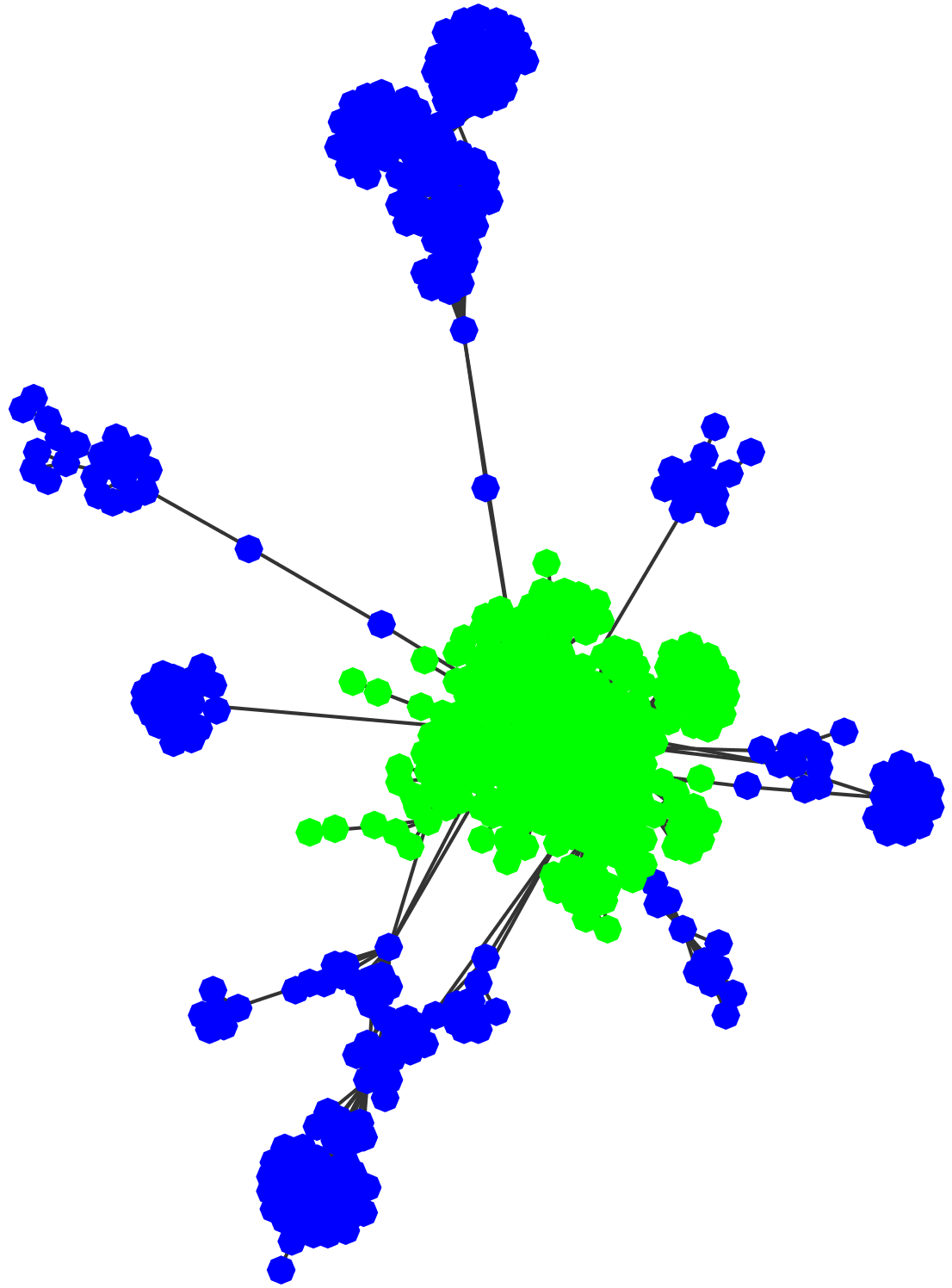}
 \end{minipage}
 \hfill
  \begin{minipage}{.15\textwidth}
\centering
\includegraphics[width=0.9\textwidth]{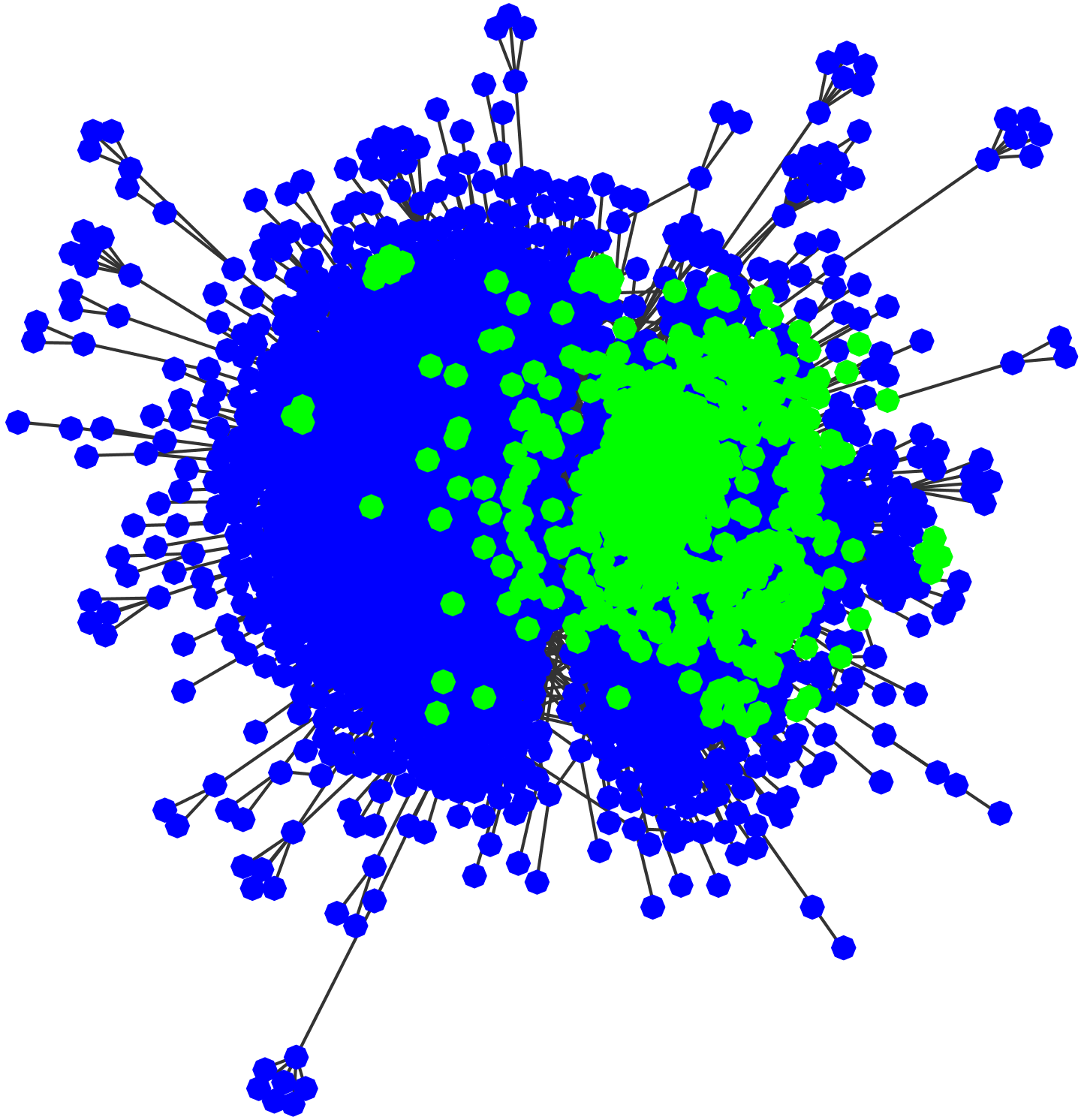}
 \end{minipage}
 \begin{minipage}{.15\textwidth}
\centering
\includegraphics[width=0.9\textwidth]{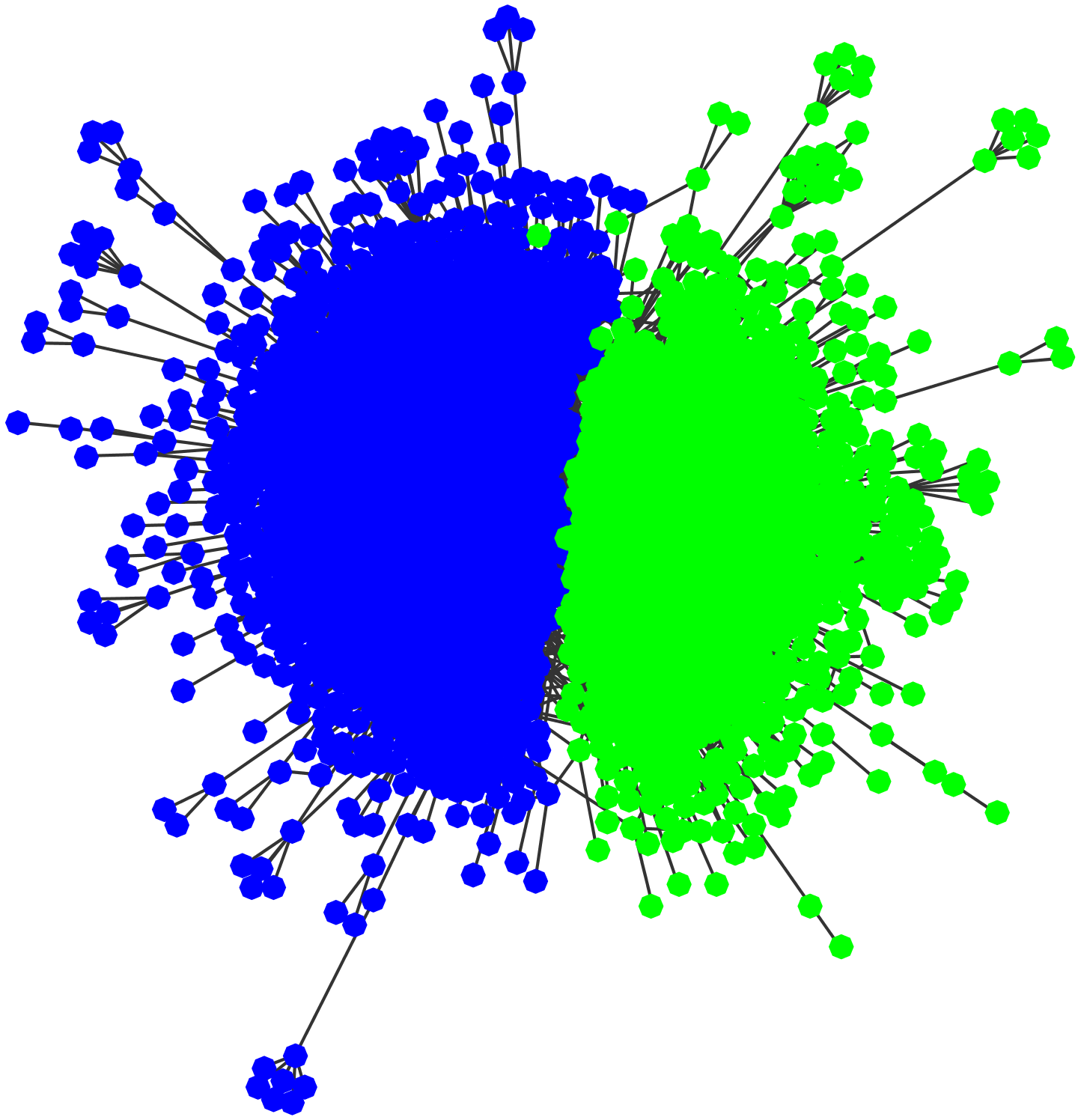}
 \end{minipage}
 \caption{Bi-partition obtained by the linear (left) and nonlinear (right) spectral methods. Networks shown, from left to right: Electronic2, Drugs, and YeastS.  }\label{fig:nobalance}
\end{figure}

\subsection{Recursive splitting for multiple communities}
The final test we propose concerns the detection of multiple communities. Although our method is meant to address the leading module problem, as in the standard spectral method, we can address multiple communities by performing Successive Graph Bipartitions (SGB). This procedure requires to update the modularity operator at each recursion, as discussed in Section \ref{sec:nonlinear_spectral_method}.   A comparison between the modularity value of the community structure obtained with different strategies on a  number of datasets is shown in Tables \ref{tab:best_kclusters} and \ref{tab:experimentsOnRealNetworks} where we compare our method with the linear spectral bi-partition and the locally greedy algorithm known as Louvain method \cite{blondel2008fast}.
For the latter method we use the GenLouvain Matlab toolbox \cite{jutla2011generalized}.

\begin{table}[t!]
 \setlength\extrarowheight{1pt}
\centering
{\footnotesize 
\begin{tabular}{!{\VRule[0.8pt]}c!{\VRule[0.8pt]}cc!{\VRule[0.8pt]}cc!{\VRule[0.8pt]}cc!{\VRule[0.8pt]}c!{\VRule[0.8pt]}}
\specialrule{0.8pt}{0pt}{0pt}        
\multirow{2}{*}{ID}&\multirow{2}{*}{Network} & \multirow{2}{*}{$n$} & \multicolumn{2}{c!{\VRule[.8pt]}}{\footnotesize{Linear Method} } &  \multicolumn{2}{c!{\VRule[.8pt]}}{\footnotesize{Nonlinear  Method}} & \multicolumn{1}{c!{\VRule[.8pt]}}{Gain (\%)}\\
& &  & $|A_1|$ & $q(A_1)$ &  $|A_2|$ & $q(A_2)$  & $q(A_2)/q(A_1)$ \\
\specialrule{.8pt}{0pt}{0pt}                                                                                                                                                                                                      
1&Macaque cortex 	        & 32   & 16   & 0.22  & 16   & 0.23 	& +4 \\ 
2&Social 3A 		& 32   & 14   & 0.28  & 17   & 0.30 	& +7\\
3&Skipwith 		& 35   & 14   & 0.04  & 17   & 0.06 	& +50 \\
4&Stony 			& 112   & 34   & 0.09 & 40   & 0.12	& +8\\
5&Malaria 		& 229   & 65   & 0.25 & 113  & 0.35 	& +40 \\
6&Electronic 2 		& 252   & 88   & 0.36 & 115  & 0.48 	& +33 \\
7&Electronic 3 		& 512   & 95   & 0.23 & 253   & 0.49 	& +113 \\
8&Drugs 			& 616   & 220  & 0.43 & 285  & 0.49 	& +14 \\
9&Transc Main 		& 662   & 91   & 0.20 & 318  & 0.44 	& +120\\
10&Software VTK		& 771   & 317  & 0.32 & 364  &	0.39    & +22 \\
11&YeastS Main 		& 2224  & 471  & 0.25 & 883  & 0.37 	& +48 \\
12&ODLIS 			& 2898  & 1285 & 0.30 & 1379 & 0.34 	& +13\\
13&Erd\H{o}s 2 		& 6927  & 1804 & 0.28 & 2333 & 0.42     & +50\\
14&AS 735 			& 7716  & 2390 & 0.18 & 3040 & 0.41     & +128\\
15&CA CondMat 		& 23133 & 2243 & 0.21 & 8777 & 0.42 	& +100\\
\specialrule{.8pt}{0pt}{0pt}           
\end{tabular}
}
\caption{Experiments on real world networks looking for two communities. 
For the nonlinear spectral method \eqref{eq:nonlinear_spectral_method}
we consider 61 starting points: 30 random, 30 diffusive (see Sec.\ \ref{sec:starting_points}) and the leading eigenvector of $M$. The column $n$ shows the size of the graph; $A_1$ and $A_2$ are the smallest communities identified by the linear and the nonlinear method, respectively; columns $|A_i|$ and $q(A_i)$ shows size and modularity value of $A_i$, $i=1,2$, respectively; the last column shows the ratio between the modularity of 
both
partitions. 
}\label{tab:experimentsOnRealNetworks_2communities}
\end{table}
\begin{figure}[t!]
\includegraphics[width=.49\textwidth]{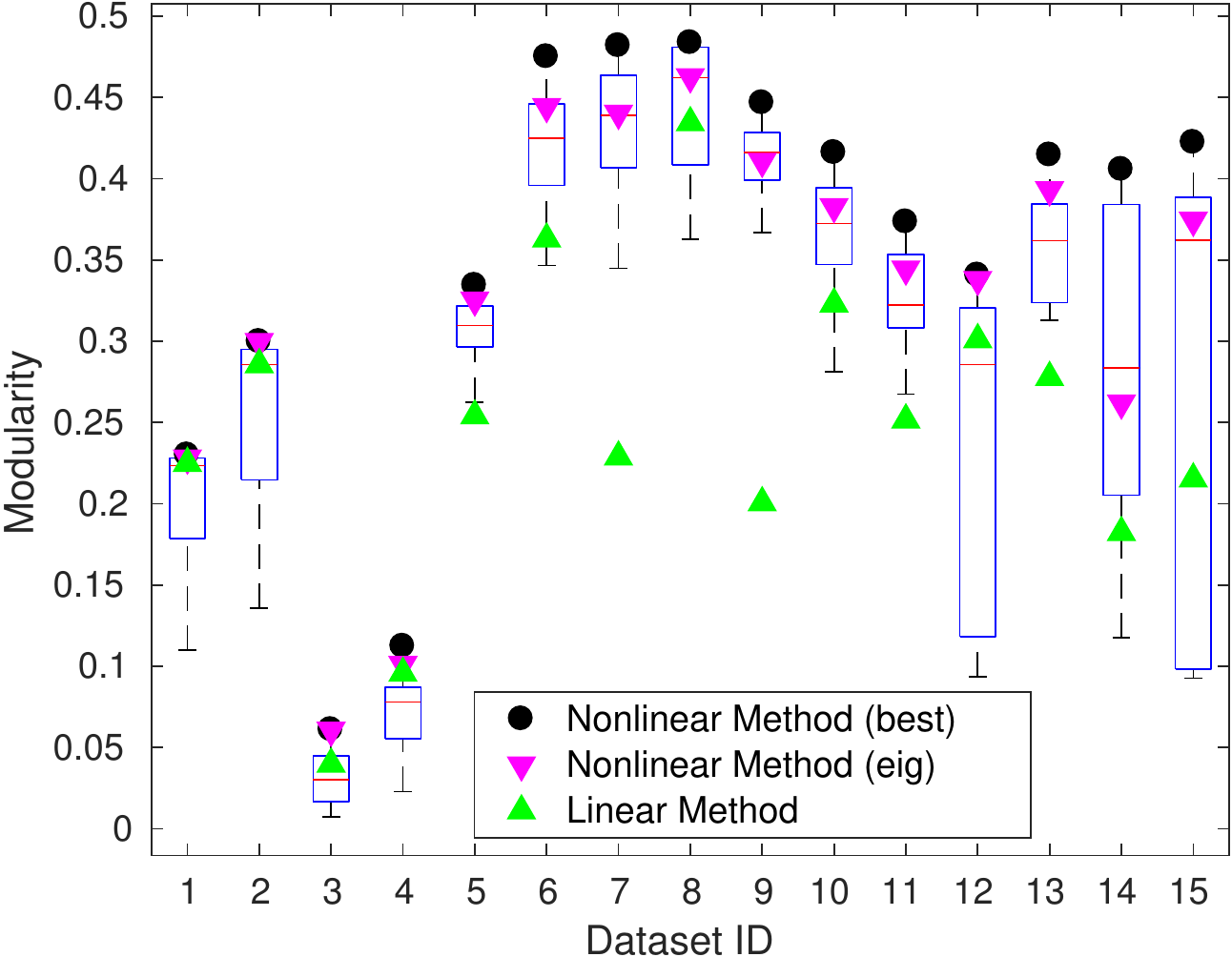}\hfill
\includegraphics[width=.47\textwidth]{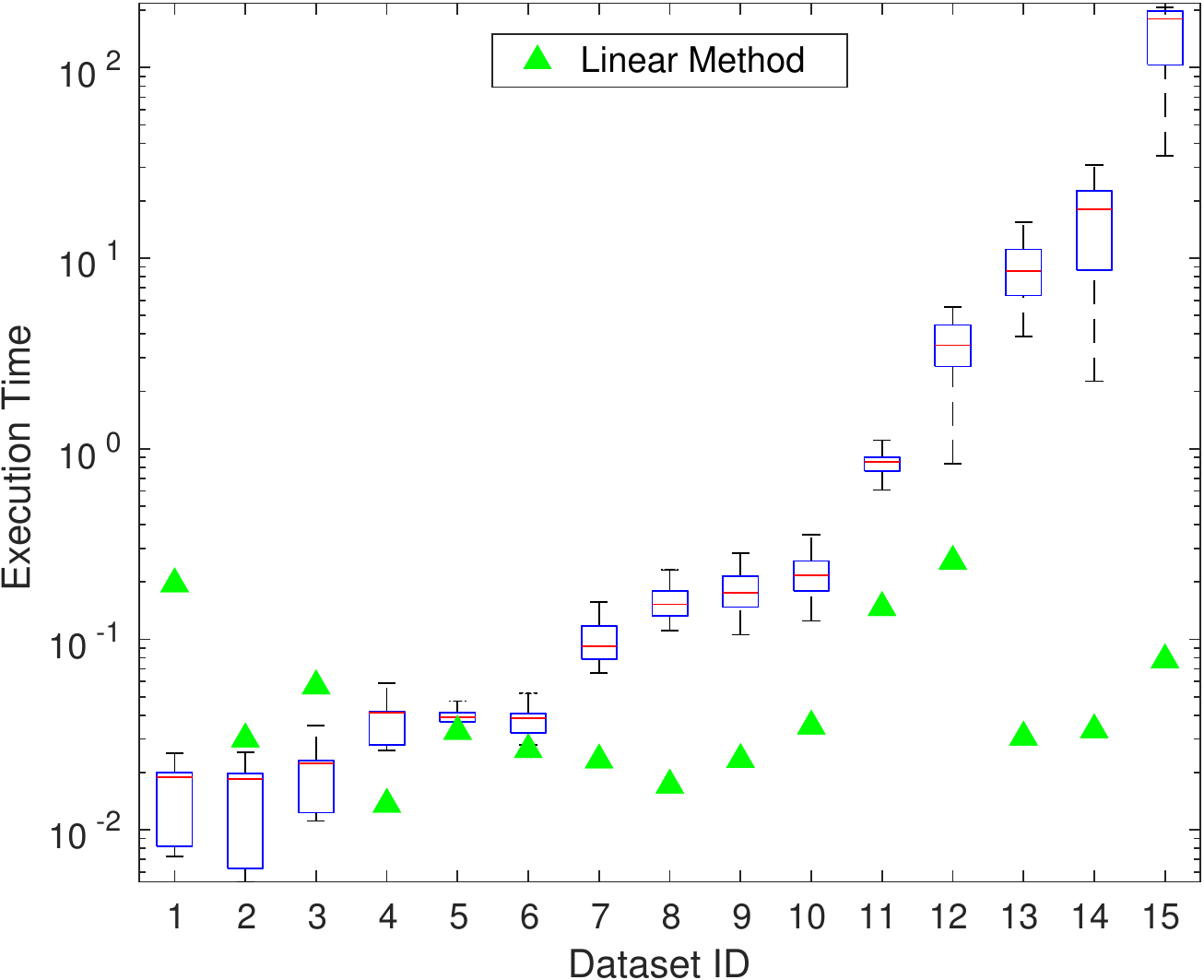}
\caption{Boxplots of modularity values (left plot) and execution times (right plot) for the nonlinear method \eqref{eq:nonlinear_spectral_method} implemented via the generalized RatioDCA Algorithm \ref{alg:genDCA} with 61 starting points: 30 random, 30 diffusive (see Sec.\ \ref{sec:starting_points}) and the leading eigenvector of $M$. 
The black dots show the best value obtained by the method (also shown in Table \ref{tab:experimentsOnRealNetworks_2communities}). The magenta triangles show values and timing for the method started with the linear modularity eigenvector. Green triangles correspond to the standard linear method. Experiments have been made on the 15 datasets of Table \ref{tab:experimentsOnRealNetworks_2communities} (where the dataset sizes are shown), with \emph{\texttt{MATLAB R2016b}} and forcing one single computing thread.}\label{fig:block-plots}
\end{figure}

These two strategies are arguably the most popular methods for revealing communities in networks. 
The 
SGB approach is a relatively naive extension of the spectral method for the leading module.
 For the modularity-based community detection problem the SGB strategy has been probably first proposed in \cite{newman2006finding}. Although this technique works well in certain cases, it typically does not outperform the Louvain method and it is known that there are situations where this approach may fail. This is shown for example in \cite{richardson2009spectral} where the method  based on the linear modularity eigenvectors is shown to fail on the 8-node bucket brigade 
 and on some real-world datasets. This negative results have led to different extensions of spectral algorithms to the problem of multiple communities, see for instance \cite{riolo2014first,zhang2015multiway}. A more careful extension of our nonlinear spectral approach to the multi-community case goes  beyond the scope of this paper and is left to future work.

In the following experiments we compute an initial community assignment via SGB and then refine it
 by moving the nodes among communities following a relatively standard flipping strategy based on the Kernighan--Lin algorithm \cite{kernighan1970efficient}, see also \cite{newman2006modularity}. This refinement procedure identifies the node that, when assigned to another community,   generates the biggest increase on modularity (or the smallest decrease if no increase
is possible). This procedure is repeated until all the $n$ nodes have been moved, with the constraint that each node assignment can be changed only once.  
By identifying the intermediate community of this process that leads to the biggest increase on modularity, the current communities are updated. Starting from these new communities, the process is repeated until no further modularity improvement is observed.

 This technique can be efficiently implemented in parallel, to speed up its time execution. We apply node flipping to both the linear and the nonlinear SGB.   

Table \ref{tab:best_kclusters} shows the percentage of cases where the nonlinear method achieves best and strictly best modularity on the 68 networks listed in Appendix \ref{app:network_names}. Table \ref{tab:experimentsOnRealNetworks} compares modularity values   and number of assigned communities on  some example networks and for  the three strategies: linear spectral method for $\lambda_1(M)$, nonlinear spectral method \eqref{eq:nonlinear_spectral_method}  for $\lambda_1(r_\M^*)$ with the generalized RatioDCA  Algorithm \ref{alg:genDCA}, and GenLouvain toolbox. 

Both our method and the Louvain method are run several times. As before, our method is run with 61 starting points: 30 random, 30 diffusive (see next subsection) and the leading eigenvector of $M$. The Louvain method is run with 100 random initial node orderings. Results in Tables \ref{tab:best_kclusters} and \ref{tab:experimentsOnRealNetworks} are based, for each method, on the best modularity assignment achieved among all the runs. As expected, the performance of our nonlinear spectral method are now less remarkable: The nonlinear method systematically outperforms the linear one, as for the leading module case discussed in the previous section, whereas it shows a performance competitive to the Louvain technique in terms of modularity value, even though the community assignment of the two methods often considerably differ. In fact, the median ratio between the modularity assignments of our method and the Louvain one over the 68 datasets of Appendix \ref{app:network_names} is 0.9998 with a variance of 0.0005.

\subsection{On the choice of the starting points}\label{sec:starting_points}
The optimization method in Algorithm \ref{alg:genDCA} often converges to local maxima, thus performances of that strategy  rely on the choice of the starting points $x_0$. According to our Theorem \ref{thm:monotony}, the sequence $r_\M^*(x_k)$ increases monotonically. This suggests that using the leading  eigenvector of the modularity matrix   as a starting point ensures a higher modularity value with respect to the linear spectral method. This observation applies to the case of two communities, whereas does not necessarily work anymore when looking for multiple groups. A standard approach in that case is to pick some additional random starting point. However a better choice can be done by choosing a set of diffuse starting points as suggested in \cite{bresson2013multiclass}: At each recursion of SGB let $\bar x$ be the eigenvector of the matrix $M_A$, corresponding to one of the current subgraphs $G(A)$. Let $v_i,v_j$ be two nodes sampled uniformly at random from $A$ such that $v_i\in C$  and $v_j \in \bar{C}$, where $\{C,\bar{C}\}$ is a partition of $A$ obtained through optimal thresholding the eigenvector $\bar x$.
Then, for the zero vector $z$ we set $z_i=1$ and $z_j=-1$. We then propagate this initial stage with $\tilde z = (I+L)^{-1}z$ where $L$ denotes the unnormalized graph Laplacian of $G(A)$, and take $\tilde z$
 as starting point for our method.

\begin{table}[t!]
 \setlength\extrarowheight{1pt}
\centering
{\footnotesize 
\begin{tabular}{!{\VRule[0.8pt]}c!{\VRule[0.8pt]}cccc!{\VRule[0.8pt]}}
\specialrule{0.8pt}{0pt}{0pt}        
Starting point strategy & Eig & 30 Rand & 30 Diff & All\\
\specialrule{.8pt}{0pt}{0pt}                                                                                                                                                                                                      
\textit{Best}			& 18\% & 29.4\% & 30.9\% & 44.1\%  \\ 
\textit{Strictly Best}   & 11.8\% & 22.1\% & 25\% & 33.8\%  \\ 
\specialrule{.8pt}{0pt}{0pt}           
\end{tabular}
}
\caption{
Experiments on real world networks looking for two or more communities. 
Fraction of cases where the nonlinear spectral method \eqref{eq:nonlinear_spectral_method} achieves best and strictly best modularity value $q$ with respect to the linear method and the best modularity value obtained by the Louvain method after 100 runs with random initial node ordering. Columns from left to right show results for different sets of starting points: linear modularity eigenvector as starting point, 30 uniformly random starting points, 30 diffused starting points (see Sec.\ \ref{sec:starting_points}), all of them. Experiments are done on 68 networks, listed in Appendix \ref{app:network_names}.}\label{tab:best_kclusters}
\end{table}
\begin{table}[t!]
\setlength\extrarowheight{1pt}
\centering
{\footnotesize 
\begin{tabular}{!{\VRule[.8pt]}cc!{\VRule[.8pt]}cc!{\VRule[.8pt]}cc!{\VRule[.8pt]}cc!{\VRule[.8pt]}cc!{\VRule[.8pt]}}
\specialrule{.8pt}{0pt}{0pt}        
\multirow{2}{*}{Network} & \multirow{2}{*}{$n$} & \multicolumn{2}{c!{\VRule[.8pt]}}{\footnotesize{Linear SGB} } &  \multicolumn{2}{c!{\VRule[.8pt]}}{\footnotesize{Nonlinear  SGB}} &  \multicolumn{2}{c!{\VRule[.8pt]}}{\footnotesize{GenLouvain}} & \multicolumn{2}{c!{\VRule[.8pt]}}{Gain (\%)}\\
&  & $q^{lin}$ & $N_c$ & $q^{nlin}$ & $N_c$ & $q^{Lou}$ & $N_c$ & $\frac{q^{nlin}}{q^{lin}}$ & $\frac{q^{nlin}}{q^{Lou}}$ \\
\specialrule{.8pt}{0pt}{0pt}                                                                                                       
Macaque cortex 	& 30 & 0.22 & 2    & 0.23 & 2    & 0.19 & 3 		& +4  &  +20 	\\
Social 3A 		& 32 & 0.36 & 4    & 0.37 & 4    & 0.37 & 4 		& +2  &  0 		\\
Skipwith 		& 35 & 0.06 & 2    & 0.07 & 2    & 0.07 & 2 		& +7 &  0 		\\
Stony 			& 112 & 0.16 & 3    & 0.17 & 5    & 0.17 & 5 	& +6 &  0 			\\
Malaria 		& 229 & 0.51 & 8    & 0.53 & 9    & 0.53 & 8	& +4 &  0 			\\
Electronic 2 	& 252 & 0.72 & 9    & 0.75 & 11   & 0.75 & 11 	& +4 &  0 		\\
Electronic 3 	& 512 & 0.76 & 25   & 0.82 & 16   & 0.79 & 15 	& +8 &  +4 	\\
Drugs 			& 616 & 0.75 & 21    & 0.77 & 17    & 0.77 & 15 	& +3 &  0 \\
Transc Main 	& 662 & 0.74 & 17    & 0.76 & 22    & 0.76 & 16 	& +4 &  0 \\
Software VTK	& 771 & 0.61 & 38    & 0.67 &	21    & 0.67 & 17  	& +12 & 0   \\
YeastS Main 	& 2224 & 0.57 & 48    & 0.59 & 46    & 0.60 & 26 	& +4 &  -2 \\
ODLIS 			& 2898 & 0.43 & 9    & 0.48 & 17    & 0.48 & 17 	& +12 &  0 \\
Erd\H{o}s 2 		& 6927 & 0.70 & 63  & 0.75 & 73   & 0.75 & 1433		& +7 &  0 \\
AS 735 			& 7716 & 0.53 & 28    & 0.63 & 77  & 0.63 & 1274 	& +19 &  0 \\
CA CondMat 		& 23133 & 0.66 & 43  & 0.72 & 832  & 0.74 & 619 	& +9 &  -3 \\
\specialrule{.8pt}{0pt}{0pt}           
\end{tabular}
}
\caption{Experiments on real world networks looking for two or more communities. 
For the Louvain method we consider 100 initial random node orderings. For the nonlinear spectral method \eqref{eq:nonlinear_spectral_method}
we consider 61 random starting points: 30 random, 30 diffusive (see Sec.\ \ref{sec:starting_points}) and the leading eigenvector of $M$. Column $n$ is the size of the graph, whereas, for each method, $N_c$ denote the number of communities identified. The three quantities $q^{lin}$, $q^{nlin}$ and $q^{Lou}$ denote the modularity of the partition obtained with the linear, nonlinear and Louvain methods, respectively. The last two columns show the ratio between the modularity of the  partitions obtained with the nonlinear method \eqref{eq:nonlinear_spectral_method} with respect the linear and the Louvain algorithms, respectively.}\label{tab:experimentsOnRealNetworks}
\end{table}

\section{Conclusions}\label{sec:comparison}
The linear spectral method \cite{newman2004finding} and the locally greedy technique known as Louvain method \cite{blondel2008fast} are among the most popular techniques for communities detection. Our nonlinear modularity approach is an extension of the linear spectral method and has a number of properties that identify it as valid alternative in several  circumstances: (a) The method is supported by a detailed mathematical understanding and two exact relaxation identities (Theorems \ref{thm:exact-1-modularity} and \ref{thm:exact-1-modularity-*}) that can be seen as nonlinear extensions of modularity Cheeger-type inequalities; (b) it exploits for the first time the use of nonlinear eigenvalue theory in the context of community detection; (c) the use of the nonlinear modularity operator $\M$, here presented, allows us to address individually both  balanced (equally sized) and unbalanced (small size)  leading module problems. 

The analysis of Section \ref{sec:experiments} shows experimental evidence of the quality of our approach and the advantage over the linear method. Several interesting research directions remain open, in particular for what concerns the computational  efficiency of the nonlinear Rayleigh quotients optimization and the overall nonlinear spectral method, and for what concerns the possibility of tailoring the method to the problem of multiple communities -- which is currently addressed by the naive strategy of successive bi-partitions -- in a more effective way.

\subsubsection*{Acknowledgements} We are grateful to Francesca Arrigo for sharing with us several of the networks we used in the numerical experiments. 

\begin{appendices}
\section{Networks used in the experiments}\label{app:network_names} Here we list the names of the networks we used in Section \ref{sec:experiments}. For the sake of brevity, we do not give individual references nor individual descriptions of the data sets, whereas we refer to \cite{estrada2015predicting, estrada2012structure, UFS, snapnets} for details. 

Network names: {\sffamily Benguela,                                                                                                                                                   
Coachella,                                                                                                                                                 
Macaque Visual Cortex Sporn,                                                                                                                                
Macaque Visual Cortex,                                                                                                                                       
PIN Afulgidus,                                                                                                                                            
Social3A,                                                                                                                                                  
Chesapeake,                                                                                                                                                
Hi-tech main,                                                                                                                                            
Zackar,                                                                                                                                                    
Skipwith,                                                                                                                                                  
Sawmill,                                                                                                                                                   
StMartin,                                                                                                                                                  
Trans urchin,                                                                                                                                             
StMarks,                                                                                                                                                   
KSHV,                                                                                                                                                    
ReefSmall,                                                                                                                                                 
Dolphins,                                                                                                                                                  
Newman dolphins,                                                                                                                                          
PRISON SymA,                                                                                                                                              
Bridge Brook,                                                                                                                                               
grassland ,                                                                                                                                                 
WorldTrade Dichot SymA,                                                                                                                                  
Shelf,                                                                                                                                                    
UKfaculty,                                                                                                                                                 
Pin Bsubtilis main,                                                                                                                                      
Ythan2,                                                                                                                                                    
Canton,                                                                                                                                                    
Stony,                                                                                                                                                    
Electronic1,                                                                                                                                               
Ythan1,                                                                                                                                                   
Software Digital main-sA,                                                                                                                               
ScotchBroom,                                                                                                                                               
ElVerde,                                                                                                                                                   
LittleRock,                                                                                                                                                
Jazz,                                                                                                                                                    
Malaria PIN main,                                                                                                                                       
PINEcoli validated main,                                                                                                                                  
SmallW main,                                                                                                                                              
Electronic2,                                                                                                                                               
Neurons,                                                                                                                                                   
ColoSpg,                                                                                                                                                   
Trans Ecoli main,                                                                                                                                        
USAir97,                                                                                                                                                   
Electronic3,                                                                                                                                               
Drugs,                                                                                                                                                    
Transc yeast main,                                                                                                                                       
Hpyroli main,                                                                                                                                             
Software VTK main-sA,                                                                                                                                   
Software XMMS main-sA,                                                                                                                                  
Roget,                                                                                                                                                   
Software Abi main-sA,                                                                                                                                   
PIN Ecoli All main,                                                                                                                                     
Software Mysql main-sA,                                                                                                                                 
Corporate People main,                                                                                                                                     
YeastS main,                                                                                                                                              
PIN Human main,                                                                                                                                          
ODLIS,                                                                                                                                                   
Internet 1997,                                                                                                                                            
Drosophila PIN Confidence main,                                                                                                                         
Internet 1998,                                                                                                                                            
Geom,                                                                                                                                                  
USpowerGrid,                                                                                                                                               
Power grid,                                                                                                                                               
Erdos02,                                                                                                                                                   
As-735,                                                                                                                                                   
Oregon1,                                                                                                                                                 
Ca-AstroPh,                                                                                                                                               
Ca-CondMat.}

\end{appendices}


\begin{thebibliography}{10}

\bibitem{arcolano2012moments}
{\sc N.~Arcolano, K.~Ni, B.~A. Miller, N.~T. Bliss, and P.~J. Wolfe}, {\em
  Moments of parameter estimates for {C}hung-{L}u random graph models}, in 2012
  IEEE International Conference on Acoustics, Speech and Signal Processing
  (ICASSP), 2012, pp.~3961--3964.

\bibitem{multires2}
{\sc A.~Arenas, A.~Fernández, and S.~Gómez}, {\em Analysis of the structure
  of complex networks at different resolution levels}, New Journal of Physics,
  10 (2008), p.~053039.

\bibitem{Bach2011}
{\sc F.~Bach}, {\em Learning with submodular functions: A convex optimization
  perspective}, Foundations and Trends in Machine Learning, 6 (2013),
  pp.~145--373.

\bibitem{beck2009fast}
{\sc A.~Beck and M.~Teboulle}, {\em A fast iterative shrinkage-thresholding
  algorithm for linear inverse problems}, SIAM Journal on Imaging Sciences, 2
  (2009), pp.~183--202.

\bibitem{blondel2008fast}
{\sc V.~D. Blondel, J.-L. Guillaume, R.~Lambiotte, and E.~Lefebvre}, {\em Fast
  unfolding of communities in large networks}, Journal of Statistical
  Mechanics: Theory and Experiment, 2008 (2008), p.~P10008.

\bibitem{boyd2017simplified}
{\sc Z.~Boyd, E.~Bae, X.-C. Tai, and A.~L. Bertozzi}, {\em Simplified energy
  landscape for modularity using total variation}, arXiv:1707.09285,  (2017).

\bibitem{brandes2008modularity}
{\sc U.~Brandes, D.~Delling, M.~Gaertler, R.~Gorke, M.~Hoefer, Z.~Nikoloski,
  and D.~Wagner}, {\em On modularity clustering}, IEEE Transactions on
  Knowledge and Data Engineering, 20 (2008), pp.~172--188.

\bibitem{bresson2013multiclass}
{\sc X.~Bresson, T.~Laurent, D.~Uminsky, and J.~von Brecht}, {\em Multiclass
  total variation clustering}, in Advances in Neural Information Processing
  Systems 26, 2013, pp.~1421--1429.

\bibitem{buhler2009spectral}
{\sc T.~B\"{u}hler and M.~Hein}, {\em Spectral clustering based on the graph
  $p$-{L}aplacian}, in Proceedings of the 26th Annual International Conference
  on Machine Learning, ICML '09, New York, NY, USA, 2009, ACM, pp.~81--88.

\bibitem{chambolle2011first}
{\sc A.~Chambolle and T.~Pock}, {\em A first-order primal-dual algorithm for
  convex problems with applications to imaging}, J. Math. Imaging Vision, 40
  (2011), pp.~120--145.

\bibitem{chang2016spectrum}
{\sc K.~C. Chang}, {\em Spectrum of the 1-{L}aplacian and {C}heeger's constant
  on graphs}, Journal of Graph Theory, 81, pp.~167--207.

\bibitem{chung2006complex}
{\sc F.~R.~K. Chung and L.~Lu}, {\em Complex graphs and networks}, vol.~107,
  American {M}athematical {S}ociety, Providence, 2006.

\bibitem{clarke1990optimization}
{\sc F.~H. Clarke}, {\em Optimization and nonsmooth analysis}, vol.~5, SIAM,
  1990.

\bibitem{clauset2004finding}
{\sc A.~Clauset, M.~E.~J. Newman, and C.~Moore}, {\em Finding community
  structure in very large networks}, Phys. Rev. E, 70 (2004), p.~066111.

\bibitem{danon2005comparing}
{\sc L.~Danon, A.~Diaz-Guilera, J.~Duch, and A.~Arenas}, {\em Comparing
  community structure identification}, Journal of Statistical Mechanics: Theory
  and Experiment, 2005 (2005), p.~P09008.

\bibitem{UFS}
{\sc T.~A. Davis and Y.~Hu}, {\em The {U}niversity of {F}lorida sparse matrix
  collection}, ACM Trans. Math. Softw., 38 (2011), pp.~1--25.

\bibitem{drabek2002generalization}
{\sc P.~Drábek and S.~B. Robinson}, {\em On the generalization of the
  {C}ourant nodal domain theorem}, Journal of Differential Equations, 181
  (2002), pp.~58 -- 71.

\bibitem{duch2005community}
{\sc J.~Duch and A.~Arenas}, {\em Community detection in complex networks using
  extremal optimization}, Phys. Rev. E, 72 (2005), p.~027104.

\bibitem{estrada2012structure}
{\sc E.~Estrada}, {\em The structure of complex networks: theory and
  applications}, Oxford University Press, 2012.

\bibitem{estrada2015predicting}
{\sc E.~Estrada and F.~Arrigo}, {\em Predicting triadic closure in networks
  using communicability distance functions}, SIAM Journal on Applied
  Mathematics, 75 (2015), pp.~1725--1744.

\bibitem{fasino2014algebraic}
{\sc D.~Fasino and F.~Tudisco}, {\em An algebraic analysis of the graph
  modularity}, SIAM Journal on Matrix Analysis and Applications, 35 (2014),
  pp.~997--1018.

\bibitem{fasino2016generalized}
{\sc D.~Fasino and F.~Tudisco}, {\em Generalized modularity matrices}, Linear
  Algebra and its Applications, 502 (2016), pp.~327 -- 345.

\bibitem{fasino2016modularity}
{\sc D.~Fasino and F.~Tudisco}, {\em Modularity bounds for clusters located by
  leading eigenvectors of the normalized modularity matrix}, Journal of
  Mathematical Inequalities, 11 (2016), pp.~701--714.

\bibitem{fortunato2010community}
{\sc S.~Fortunato}, {\em Community detection in graphs}, Physics reports, 486
  (2010), pp.~75--174.

\bibitem{fortunato2007resolution}
{\sc S.~Fortunato and M.~Barth{\'e}lemy}, {\em Resolution limit in community
  detection}, Proceedings of the National Academy of Sciences, 104 (2007),
  pp.~36--41.

\bibitem{fred2006learning}
{\sc A.~L.~N. Fred and A.~K. Jain}, {\em Learning pairwise similarity for data
  clustering}, in 18th International Conference on Pattern Recognition (ICPR),
  2006, pp.~925--928.

\bibitem{gleiser2003community}
{\sc P.~M. Gleiser and L.~Danon}, {\em Community structure in jazz}, Advances
  in Complex Systems, 06 (2003), pp.~565--573.

\bibitem{guimera2004modularity}
{\sc R.~Guimer\`a, M.~Sales-Pardo, and L.~A.~N. Amaral}, {\em Modularity from
  fluctuations in random graphs and complex networks}, Phys. Rev. E, 70 (2004),
  p.~025101.

\bibitem{hein2010inverse}
{\sc M.~Hein and T.~B{\"u}hler}, {\em An inverse power method for nonlinear
  eigenproblems with applications in 1-spectral clustering and sparse {PCA}},
  in Adv. Neural Inf. Process. Syst. 23 (NIPS), 2010, pp.~847--855.

\bibitem{hein2011beyond}
{\sc M.~Hein and S.~Setzer}, {\em Beyond spectral clustering - tight
  relaxations of balanced graph cuts}, in Advances in Neural Information
  Processing Systems 24, 2011, pp.~2366--2374.

\bibitem{hiriart2012fundamentals}
{\sc J.-B. Hiriart-Urruty and C.~Lemar{\'e}chal}, {\em Fundamentals of convex
  analysis}, Springer Science \& Business Media, 2012.

\bibitem{hu2013method}
{\sc H.~Hu, T.~Laurent, M.~A. Porter, and A.~L. Bertozzi}, {\em A method based
  on total variation for network modularity optimization using the {MBO}
  scheme}, SIAM Journal on Applied Mathematics, 73 (2013), pp.~2224--2246.

\bibitem{jutla2011generalized}
{\sc L.~G.~S. Jeub, M.~Bazzi, I.~S. Jutla, and P.~J. Mucha}, {\em A generalized
  {L}ouvain method for community detection implemented in {MATLAB}},
  (2011--16), \url{http://netwiki.amath.unc.edu/GenLouvain}.

\bibitem{kamada1989algorithm}
{\sc T.~Kamada and S.~Kawai}, {\em An algorithm for drawing general undirected
  graphs}, Information Processing Letters, 31 (1989), pp.~7 -- 15.

\bibitem{kernighan1970efficient}
{\sc B.~W. Kernighan and S.~Lin}, {\em An efficient heuristic procedure for
  partitioning graphs}, The Bell System Technical Journal, 49 (1970),
  pp.~291--307.

\bibitem{lancichinetti2009community}
{\sc A.~Lancichinetti and S.~Fortunato}, {\em Community detection algorithms: A
  comparative analysis}, Phys. Rev. E, 80 (2009), p.~056117.

\bibitem{lancichinetti2011limits}
{\sc A.~Lancichinetti and S.~Fortunato}, {\em Limits of modularity maximization
  in community detection}, Phys. Rev. E, 84 (2011), p.~066122.

\bibitem{lecun1998mnist}
{\sc Y.~LeCun, C.~Cortes, and C.~J.~C. Burges}, {\em The {MNIST} database of
  handwritten digits},  (1998), \url{http://yann.lecun.com/exdb/mnist/}.

\bibitem{snapnets}
{\sc J.~Leskovec and A.~Krevl}, {\em {SNAP Datasets}: {Stanford} large network
  dataset collection}.
\newblock \url{http://snap.stanford.edu/data}, June 2014.

\bibitem{pmlr-v84-mercado18a}
{\sc P.~Mercado, A.~Gautier, F.~Tudisco, and M.~Hein}, {\em The power mean
  {L}aplacian for multilayer graph clustering}, in Proceedings of the 21st
  International Conference on Artificial Intelligence and Statistics (AISTATS),
  vol.~84 of Proceedings of Machine Learning Research, 2018, pp.~1828--1838.

\bibitem{newman2006finding}
{\sc M.~E.~J. Newman}, {\em Finding community structure in networks using the
  eigenvectors of matrices}, Phys. Rev. E, 74 (2006), p.~036104.

\bibitem{newman2006modularity}
{\sc M.~E.~J. Newman}, {\em Modularity and community structure in networks},
  Proceedings of the National Academy of Sciences, 103 (2006), pp.~8577--8582.

\bibitem{newman2010networks}
{\sc M.~E.~J. Newman}, {\em Networks: an introduction}, Oxford {U}niversity
  {P}ress, 2010.

\bibitem{newman2004finding}
{\sc M.~E.~J. Newman and M.~Girvan}, {\em Finding and evaluating community
  structure in networks}, Phys. Rev. E, 69 (2004), p.~026113.

\bibitem{porter2009communities}
{\sc M.~A. Porter, J.-P. Onnela, and P.~J. Mucha}, {\em Communities in
  networks}, Notices of the AMS, 56 (2009), pp.~1082--1097.

\bibitem{prvzulj2006modelling}
{\sc N.~Pr{\v{z}}ulj and D.~J. Higham}, {\em Modelling protein--protein
  interaction networks via a stickiness index}, Journal of The Royal Society
  Interface, 3 (2006), pp.~711--716.

\bibitem{reichardt2006statistical}
{\sc J.~Reichardt and S.~Bornholdt}, {\em Statistical mechanics of community
  detection}, Phys. Rev. E, 74 (2006), p.~016110.

\bibitem{richardson2009spectral}
{\sc T.~Richardson, P.~J. Mucha, and M.~A. Porter}, {\em Spectral
  tripartitioning of networks}, Phys. Rev. E, 80 (2009), p.~036111.

\bibitem{riolo2014first}
{\sc M.~A. Riolo and M.~E.~J. Newman}, {\em First-principles multiway spectral
  partitioning of graphs}, Journal of Complex Networks, 2 (2014), pp.~121--140.

\bibitem{multires3}
{\sc P.~Ronhovde and Z.~Nussinov}, {\em Local resolution-limit-free potts model
  for community detection}, Phys. Rev. E, 81 (2010), p.~046114.

\bibitem{schaeffer2007graph}
{\sc S.~E. Schaeffer}, {\em Graph clustering}, Computer science review, 1
  (2007), pp.~27--64.

\bibitem{shen2010spectral}
{\sc H.-W. Shen and X.-Q. Cheng}, {\em Spectral methods for the detection of
  network community structure: a comparative analysis}, Journal of Statistical
  Mechanics: Theory and Experiment, 2010 (2010), p.~P10020.

\bibitem{traag2011narrow}
{\sc V.~A. Traag, P.~Van~Dooren, and Y.~Nesterov}, {\em Narrow scope for
  resolution-limit-free community detection}, Phys. Rev. E, 84 (2011),
  p.~016114.

\bibitem{traud2011comparing}
{\sc A.~Traud, E.~Kelsic, P.~Mucha, and M.~Porter}, {\em Comparing community
  structure to characteristics in online collegiate social networks}, SIAM
  Review, 53 (2011), pp.~526--543.

\bibitem{tudisco2016nodal}
{\sc F.~Tudisco and M.~Hein}, {\em A nodal domain theorem and a higher-order
  {C}heeger inequality for the graph $p$-{L}aplacian}, EMS Journal of Spectral
  Theory, in press (2016).

\bibitem{tudisco2018core}
{\sc F.~Tudisco and D.~J. Higham}, {\em A nonlinear spectral method for
  core-periphery detection in networks}, arXiv:1804.09820,  (2018).

\bibitem{zhang2013normalized}
{\sc S.~Zhang and H.~Zhao}, {\em Normalized modularity optimization method for
  community identification with degree adjustment}, Phys. Rev. E, 88 (2013),
  p.~052802.

\bibitem{zhang2015multiway}
{\sc X.~Zhang and M.~E.~J. Newman}, {\em Multiway spectral community detection
  in networks}, Phys. Rev. E, 92 (2015), p.~052808.

\end{thebibliography}
\end{document}